%% file: iclr2024_conference.tex
\documentclass{article} 
\usepackage{iclr2024_conference,times}
\usepackage{marvosym}
\input{math_commands.tex}

\usepackage{hyperref}
\usepackage{url}
\usepackage{colortbl} 
\usepackage{textcomp}

\usepackage{listings}
\usepackage{tcolorbox}
\usepackage[utf8]{inputenc} 
\usepackage[T1]{fontenc}    
\usepackage{hyperref}[hyperfootnotes=false]       
\usepackage{url}            
\usepackage{booktabs}       
\usepackage{amsfonts}       
\usepackage{nicefrac}       
\usepackage{microtype}      
\usepackage{xcolor}         
\usepackage{graphicx}

\usepackage{bbding}
\usepackage[commandnameprefix=always]{changes}
\newcommand{\angstrom}{\textup{\AA}}
\usepackage{threeparttable}
\usepackage{tabularx}
\usepackage{multirow}

\usepackage{microtype}
\usepackage{amsthm}
\usepackage{graphicx}
\usepackage{subfigure}
\usepackage{booktabs} 

\usepackage{amsmath}
\newtheorem{myDef}{Definition}

\usepackage{amssymb}
\usepackage[ruled,vlined,linesnumbered]{algorithm2e}
\usepackage{graphicx}
\newtheorem{theorem}{Theorem}[section]
\newtheorem{proposition}[theorem]{Proposition}

\theoremstyle{definition}

\theoremstyle{remark}

\usepackage{appendix}
\usepackage{titletoc}
\newcommand\DoToC{%
  \startcontents
  \printcontents{}{1}{\textbf{Contents}\vskip3pt\hrule\vskip5pt}
  \vskip3pt\hrule\vskip5pt
}

\definecolor{Gray}{gray}{0.9} 

\title{Long-Short-Range Message-Passing: \\ A Physics-Informed Framework to Capture Non-Local Interaction for Scalable Molecular Dynamics Simulation}

\author{%
Yunyang Li $^{\spadesuit}$\thanks{Work done during internship at Microsoft.}$\ $ \thanks{Equal contribution.}  \quad Yusong Wang$^{\clubsuit \heartsuit}$\footnotemark[2] \quad Lin Huang$^{\clubsuit}$\textsuperscript{\Letter} \quad Han Yang$^{\clubsuit}$ \quad Xinran Wei$^{\clubsuit}$ \\
\textbf{Jia Zhang}$^{\clubsuit}$\textsuperscript{\Letter} \quad \textbf{Tong Wang}$^{\clubsuit}$\textsuperscript{\Letter} \quad \textbf{Zun Wang}$^{\clubsuit}$ \quad \textbf{Bin Shao}$^{\clubsuit}$ \quad \textbf{Tie-Yan Liu}$^{\clubsuit}$\\
$^{\spadesuit}$Yale University  \quad $^{\clubsuit}$MSR AI4Science \quad $^{\heartsuit}$Xi’an Jiaotong University \\
\textsuperscript{\Letter}\texttt{\{huang.lin, jia.zhang, watong\}@microsoft.com}\\
}
\iclrfinalcopy

\begin{document}

\maketitle

\begin{abstract}
Computational simulation of chemical and biological systems using \textit{ab initio} molecular dynamics has been a challenge over decades.
Researchers have attempted to address the problem with machine learning and fragmentation-based methods. 
However, the two approaches fail to give a satisfactory description of long-range and many-body interactions, respectively.
Inspired by fragmentation-based methods, we propose the Long-Short-Range Message-Passing (LSR-MP) framework as a generalization of the existing equivariant graph neural networks (EGNNs) with the intent to incorporate long-range interactions efficiently and effectively. 
We apply the LSR-MP framework to the recently proposed  ViSNet and demonstrate the state-of-the-art results with up to 40\% MAE reduction for molecules in MD22 and Chignolin datasets. 
Consistent improvements to various EGNNs will also be discussed to illustrate the general applicability and robustness of our LSR-MP framework. The code for our experiments and trained model weights could be found at \url{https://github.com/liyy2/LSR-MP}.

\end{abstract}

\section{Introduction}
\emph{Ab initio} molecular dynamics (AIMD)~\cite{car1985unified} has been an indispensable tool in the fields of chemistry, biology, and material science. By virtue of its effective description of the kinetic and thermodynamic properties in molecular and condensed systems, AIMD is capable of  elucidating numerous phenomena of interest, such as chemical reactions~\cite{hwang2015reaction}, protein folding~\cite{cellmer2011making} and electron-phonon interactions~\cite{karsai2018electron, kundu2021quantum, kundu2022influence}. However, AIMD simulations driven by conventional quantum chemical methods such as density functional theory (DFT)~\cite{hohenberg1964inhomogeneous, kohn1965self} become prohibitively expensive as the size of the system increases, due to the high computational cost~\cite{szabo2012modern}. 

Computational chemists have sought to address the challenges associated with understanding molecular interactions by adopting a divide-and-conquer strategy. This approach has significantly spurred the evolution of fragmentation-based methods in recent years~\cite{yang1991direct, exner2003ab, he2005new, he2006generalized, li2007generalized, he2014fragment}. Central to these methods is the principle of chemical locality~\cite{he2014fragment}, which posits that chemical subsystems often exhibit minimal or weak interactions with one another. This principle is grounded in the broader hypothesis that molecular behaviors and interactions largely arise from localized phenomena. Empirical evidence further supports this notion~\cite{collins2015energy}, demonstrating that molecular reactivity is primarily dictated by the presence, composition, and spatial configuration of specific functional groups. Consequently, by subdividing large molecules into smaller, more manageable subsystems, one can facilitate parallelized studies and gain deeper insights into their intricate behaviors.
Following the computation on these subsystems, there are \emph{manually designed} techniques available to reconstitute the holistic properties of the entire molecule. Notable examples of these assembly methods are rooted in energy- and density matrix-based approaches~\cite{li2007generalized, li2021computational}.
Regardless of the method selected, the many-body interactions among subsystems must be carefully calibrated, 
as this is the primary source of error in the fragmentation-based method. A common practice is to take into account the two- and/or three-body interactions.
However, such approximation compromises the accuracy of the fragmentation-based methods and thus has significantly limited their applicability~\cite{li2021computational}.

Recent advances in the application of machine learning in chemistry~\cite{zhang2018deep, behler2021four, unke2021machine, deringer2021gaussian, zhang2022deep} have facilitated AIMD simulations of larger molecules without significant loss of accuracy, in contrast to the prohibitive computational costs of conventional DFT calculations.
In particular, the recent equivariant graph neural networks (EGNNs) introduce the inductive bias of symmetry into the model~\cite{han2022geometrically}, which further improves the model's capacity to represent the geometric structures of the molecular systems.
Existing EGNNs model interactions in local environments by incorporating structural information such as bond lengths, bond angles, and dihedral angles within a predefined radius cutoff. Although being effective for datasets consisting of small molecular systems, e.g. MD17~\cite{chmiela2018towards, schutt2017quantum, chmiela2017machine} and QM9~\cite{ramakrishnan2014quantum}, this approach results in substantial loss of information for larger molecules in which long-range interactions, including electrostatic and van der Waals forces, are non-negligible.
Increasing the radius cutoff and stacking more layers are common methods to remedy the information loss, but they inevitably bring about new issues, e.g. loss of data efficiency, convergence difficulty, and information over-squashing~\cite{alon2020bottleneck}.

In light of these challenges, we introduce the fragmentation-based message passing on top of the current EGNNs, called the \emph{Long-Short-Range Message-Passing} (LSR-MP) framework, with the objective to (1) explicitly incorporate long-range interactions, (2) compensate for the loss of higher-order many-body interaction in existing fragmentation-based approaches and (3) maintain computational efficiency. 
As illustrated in Fig.~\ref{icml-historical2}, the input molecule is processed at two different levels.
In the first level, we introduce the \textit{short-range} module to encode the many-body interactions under the local neighborhood.
In the second level, we formulate the \textit{fragmentation} module and the \textit{long-range} module, which are dedicated to modeling long-range interactions between fragments and atoms.  We compare the key differences of these methodologies in Table \ref{tab:diff}. Our contribution could be listed as follows:
\begin{enumerate}
    \item We introduce LSR-MP, a novel message-passing framework for neural network-based potential modeling. Leveraging BRICS fragmentation techniques and long-range message-passing, LSR-MP effectively captures long-range information, showcasing its robust performance in modeling large molecular systems.
    \item  We present an exemplary implementation of the LSR-MP framework, dubbed ViSNet-LSRM. Our approach leverages the interaction between vectorial and scalar embeddings to maintain equivariance~\cite{wang2022ViSNet, wang2022ensemble}. Notably, ViSNet-LSRM achieves competitive performance on large molecular datasets such as MD22~\cite{doi:10.1126/sciadv.adf0873} and Chignolin~\cite{wang2022ViSNet}, while utilizing fewer parameters and offering up to 2.3x faster efficiency.
    \item  We illustrate the general applicability of the LSR-MP framework using other EGNN backbones, such as Equiformer, PaiNN, and ET, which results in substantial performance improvements over the original models. This highlights the framework's superior adaptability and effectiveness.
\end{enumerate}
\begin{table}[]
\centering
\caption{Comparison of the methods. $N$ number of atoms, $m$ number of basis.}
\label{tab:diff}
\resizebox{0.8\columnwidth}{!}{%
\begin{tabular}{@{}lllll@{}}
\toprule
\textbf{Method}                   & \textbf{Efficiency}                 & \textbf{Many-body } & \textbf{Long-range } & \textbf{Accuracy} \\ \midrule
\textit{Ab initio} Methods                         & $O(N^a m^b) [a \geq 3, m = 2\sim4]$                            & \CheckmarkBold                            & \CheckmarkBold                            & Most Accurate              \\
Fragmentation Methods       & $O(N) + O(N_{frag} (N/N_{frag})^a m^b) [a \geq 3, m = 2\sim4] $ & \XSolidBrush                               & \CheckmarkBold                             & 1 - 5 kcal/mol               \\
EGNN                        & $O(N^a) [a = 1 \sim 2]$                              & \CheckmarkBold                             & \XSolidBrush                           & Optimally < 1 kcal/mol                \\
EGNN + LSR-MP & $O(N^a) [a = 1 \sim 2]$                              & \CheckmarkBold                             & \CheckmarkBold                              & Optimally < 1 kcal/mol               \\ \bottomrule
\end{tabular}%
}
\end{table}

\section{Related Work}
\subsection{Fragmentation-Based Methods}
Fragmentation techniques, historically employed to facilitate quantum mechanical (QM) computations for expansive molecular systems, present solutions to intricate scaling issues~\cite{ganesh2006molecular, li2007generalized}. In recent research trajectories, there has been a discernible bifurcation into two prominent fragmentation methodologies: density matrix-based and energy-based.
Within the density matrix paradigm, the overarching strategy entails the construction of an aggregate density matrix, sourced from the individual matrices of subsystems. Subsequent property extractions hinge upon this consolidated matrix. Noteworthy implementations in this domain encompass the divide-and-conquer (DC) technique~\cite{yang1991direct}, the adjustable density matrix assembler (ADMA) framework~\cite{exner2003ab}, the elongation approach~\cite{gu2004new}, and the molecular fractionation with conjugated caps (MFCC) mechanism~\cite{he2005new}.
Conversely, energy-based fragmentation methodologies pivot on a direct assimilation of subsystem properties to deduce the global property. A seminal contribution by He \emph{et al.} involved the adaptation of the MFCC model into a more streamlined formulation~\cite{he2006generalized}, which was subsequently enhanced with electrostatic embedding schemes, capturing long-range interactions~\cite{he2014fragment}. Parallel advancements in the energy-based arena include the molecular tailoring approach~\cite{gadre2006molecular}, the molecules-in-molecules (MIM) paradigm~\cite{mayhall2011molecules}, and the systematic fragmentation methodology~\cite{deev2005approximate}.

\subsection{Equivariant Graph Neural Networks}
Equivariant Graph Neural Networks (EGNNs) inherently weave symmetry's inductive bias into model architectures. Traditional EGNNs predominantly lean on group theory, invoking irreducible representations~\cite{thomas2018tensor, fuchs2020se, anderson2019cormorant} and leveraging the Clebsch-Gordan (CG) product to guarantee equivariance. Contemporary strides in this domain, such as those presented in~\cite{batzner20223, frank2022sokrates, batatia2022mace}, have incorporated higher-order spherical harmonic tensors, resulting in notable performance enhancements. Parallel to these, alternative advancements have been steered by methods like PaiNN~\cite{schutt2021equivariant} and TorchMD-NET~\cite{tholke2022torchmd}, which prioritize iterative updates of scalar and vectorial features. ViSNet~\cite{wang2022ViSNet} builds upon PaiNN's foundation, integrating runtime geometric computation (RGC) and vector-scalar interactive message passing (ViS-MP). 

\section{Long-Short-Range Message-Passing}

\textbf{Notations:}
We use $h$ and $\vec{v}$ to represent short-range scalar and vectorial embeddings, respectively, and $x$ and $\vec{\mu}$ for their long-range counterparts. Capitalized fragment embeddings are signified by $H$ and $\vec{V}$. Subscripts align with atom or fragment indexes, while superscripts denote layers in a multi-layered network. Additionally, we use $\text{dist}(\cdot)$ for the Euclidean distance, $\odot$ for the Hadamard product, and $\langle\cdot, \cdot\rangle$ for vector scalar products. $\mathbf{1}^d$ is a dimension $d$ column vector of ones. $L_\mathrm{short}$ the number of short-range layer, $L_\mathrm{long}$ the number of long-range layer. Functions include:  $\textsc{Dense}(\cdot)$ for an activated linear layer with bias, $\textsc{Linear}(\cdot)$ for a biased linear layer, and $U(\cdot)$ for an unbiased linear layer. Network parameters are typically unshared unless specified.

\begin{figure*}[t]
\begin{center}
\centerline{\includegraphics[width= \columnwidth]{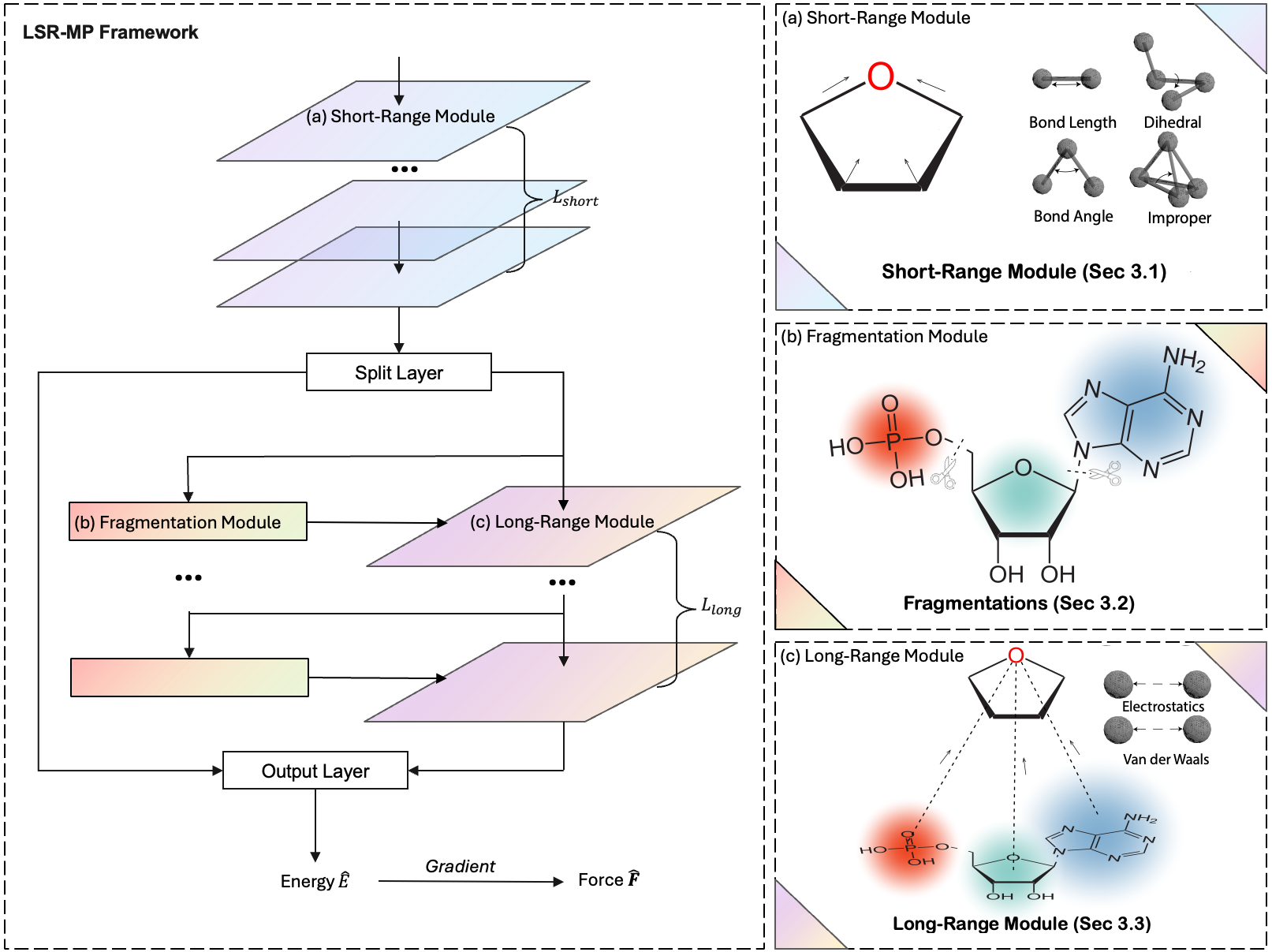}}
\vspace{-0.12in}
\caption{LSR-MP Framework:
(a) A center atom (in red) uses short-range equivariant message passing on a radius graph for its embedding. A linear projection and split process in the split layer channel some info to the fragmentation and long-range modules.
(b) The fragmentation module dissects the system and derives fragment representations.
(c) The long-range module facilitates bipartite equivariant message passing between atoms and fragments. The output layer merges both long and short-range embeddings, then sums over the graph for graph-level embeddings.}
\label{icml-historical2}
\vspace{-10mm}
\end{center}
\end{figure*}

\subsection{Short-Range Module}
Considering a short-range 
 radius graph $\mathcal{G_\text{short}}$ with node set $\mathcal{V}$ for atoms and edge set $\mathcal{E_\text{short}}$ which is defined as $
        \mathcal{E} _{\text{short}} = \left\{e_{ij} | \text{dist}(\vec{p}_i, \vec{p}_j) \leq r_{\text{short}},\ \forall i,j\in \mathcal{V} \right\}.
$ 
The short-range module performs message passing on $\mathcal{G}_{\text{short}}$ by taking the atomic numbers $Z \in \mathbb{N}^{n \times 1}$ and positions $\vec{p} \in \mathbb{R}^{n \times3}$ ($n$ is the number of atoms in the system) as input, and is targeted to model the geometric information (bonds, angle, dihedral, improper) on $\mathcal{G_{\text{short}}}$. 

For illustration purposes, we suppose the short-range module operates on scalar embeddings and vectorial embeddings, while generalizations to higher-order spherical harmonics embeddings could be easily derived with the Clebsch-Gordan tensor product, and is included in Appendix~\ref{higher-order}.

In general, the short-range module adopts an iterative short-range message-passing scheme:
{\small
\begin{equation}
        h_i^{l}, \vec{v}_i^{\,l} =  \phi_{\text{Short}}\left(h_i^{l-1}, \vec{v}_i^{\,l-1}, \sum_{j\in N(i)} m
_{ij}^{l-1}, \sum_{j\in N(i)}\vec{m}_{ij}^{l-1}\right). 
\end{equation}
}The $\phi_{\text{Short}}(\cdot)$ defines a message-passing framework, and $N(\cdot)$ denotes the first-order neighbor set of a particular node. $m_{ij}$ and $\vec{m}_{ij}$ denote the scalar message and vectorial message between node $i$ and its first order neighbor $j$. 

\subsection{Fragmentation Module}

To incorporate long-range components, we look towards the fragmentation-based methods prevalent in quantum chemistry. While fragmentation-based methods might not adept at depicting many-body interactions, they can potentially capture long-range interactions with good accuracy. For fragmentation-based methods, the preservation of intrinsic chemical structures during fragmentation is paramount, as this would minimize the artifacts induced by fragmentation-based methods. The BRICS method, as implemented in RDKit~\cite{landrum2013rdkit}, offers a robust solution to this challenge. Formally, the BRICS fragmentation procedure is defined as: $S = \textsc{BRICS}(Z, \vec{p}),$ where $S\in \{0,1\}^{n\times m}$, $n$ is the number of atoms, $m$ is the number of fragments.

BRICS utilizes 16 predefined chemical motifs, selected based on organic chemistry principles, to direct the fragmentation process towards meaningful substructures. It aims to reduce energy discrepancies arising from bond breakages, ensuring a chemically appropriate local environment. Additionally, BRICS applies post-fragmentation techniques: removing small fragments, handling duplicates, and correcting overlapping fragment issues. Detailed implementation of BRICS in the fragmentation module is provided in Appendix \ref{app-brics}.

\textbf{Fragment Representation Learning:} To obtain effective representations, the fragment learning function summarized the atom's embeddings within fragments.  
We denote an arbitrary atom scalar embedding as $\mathbf{h}$, vectorial embeddings as $\vec{\mathbf{v}}$ \footnote{$\mathbf{h},\vec{\mathbf{v}}$ could be short-range embeddings $h, \vec{v}$, or long-range embeddings $x, \vec{\mu}$ in sec \ref{sec:Long-Range Module}.}. 
Analogous to atoms, the fragment's scalar embedding is invariant under SO(3) transformations, and its vectorial representation is equivariant to SO(3) transformations. Characterization of vectorial fragment embeddings enables the network to model equivariant long-range interactions. The general learning scheme is defined as follows:
\begin{eqnarray}
     \label{eq:FLfunction1}
    H_j^{l} = \sum_{i\in S(j)}  \alpha_i^{l} \odot \mathbf{h}_i^{l}; \ \ 
    \vec{V}_j^{l} = \sum_{ i\in S(j)} \beta_i^{l} \odot \vec{\mathbf{v}}_i^{l}; \ \ 
    \vec{P}_j = \sum_{ i\in S(j)} \gamma_i \vec{p_i} + \kappa_i^{l} \odot \vec{\mathbf{v}}_i^{l}; \ ,
\end{eqnarray}
in which $H_j^{l}$, $\vec{V}_j^{l}$ and $\Vec{P}_{j}^{l}$ denotes the scalar embedding, vectorial embedding, and position of fragment $j$, respectively. $j$ is the index for fragments, and $S(j)$ is the set induced by the assignments of fragment $j$.  $\alpha_i^{l}, \beta_i^{l},  \kappa_i^{l}\in R^d$ are weight vectors for each atom within the fragments and should be SO(3) invariant. Additionally, to guarantee translational equivariance, $\gamma_i$ are parameterized so that $\sum_{i\in S(j)} \gamma_i = 1, \gamma_i \geq 0 $.

\subsection{Long-Range Module: Geometric Bipartite Message Passing}
\label{sec:Long-Range Module}
Considering a bipartite radius graph $\mathcal{G_\text{long}}$ with node set $\{\mathcal{V,U}\}. $ $\mathcal{V}$ is the atoms set and $\mathcal{U}$ is the fragments set, and the edge set is defined as: 
{\small
 \begin{equation}
        \mathcal{E}_{\text{long}} = \left\{e_{ij} \,\big |\, \text{dist}\left(\vec{p}_i, \vec{P}_j\right) \leq r_{\text{long}},\,\forall i\in \mathcal{U},\, \forall j \in \mathcal{V}\right\}.
\end{equation}
}$r_{\text{long}}$ is the long-range cutoff and is normally chosen to be much larger than the radius of the short-range neighbor graph $r_{\text{short}}$.
The bipartite geometric message-passing is performed to characterize long-range interactions:
{\small
\begin{equation}
        x^l_i, \vec{\mu}_i^{l} =  \psi_{\text{long}}\left(x_i^{l-1}, \vec{\mu}_i^{l-1}, \sum_{j\in N(i)} M
_{ij}^{l-1}, \sum_{j\in N(i)}\vec{M}_{ij}^{l-1}\right). 
\end{equation}
}$\psi(\cdot)_{\text{long}}$ is the general bipartite message passing framework, $x^l_i$ is the long-range scalar embedding, $\vec{\mu_i}^l$ is the long-range vectorial embedding, $i,j$ are index for atom, fragment respectively, and $N(\cdot)$ is the neighborhood of atom $i$ on the atom-fragment bipartite graph. $M_{ij}$ and $\vec{M}_{ij}$ denote the bipartite scalar message and vectorial message between atom $i$ and its incident fragment $j$.

\subsection{Properties Prediction}
The short-range embeddings, denoted as $h$ and $\vec{v}$, likely capture local interactions, while the long-range embeddings, represented by $x$ and $\vec{\mu}$, encompass long-range interactions. These two types of embeddings supplement each other, offering a comprehensive modeling of the systems. To effectively combine these embeddings, LSR-MP employs a late fusion strategy:
\begin{eqnarray}
        h_{\text{out}} &=& \textsc{Dense}\left(\left[h^{L_\text{short}}, x^{L_\text{long}}\right]\right),  \\ \vec{v}_{\text{out}} &=& U\left(\left[\vec{v}^{L_\text{short}}, \vec{\mu}^{L_\text{long}}\right]\right).
\end{eqnarray}
$L_\text{short}$ and $L_\text{long}$ are the number of layers for the short-range module and long-range module respectively. The output layer predicts scalar properties using the scalar embedding
$h_{\text{out}}$ or predict tensorial properties using $\vec{v}_{\text{out}}$ and $h_{\text{out}}$.

\section{ViSNet-LSRM}
Based on LSR-MP framework, we provided an exemplary implementation called ViSNet-LSRM. It uses ViSNet~\cite{wang2022ViSNet} as the backbone for the short-range module, thus $h^{l}, \vec{v}^{l} =  \textsc{ViSNet}(Z, \vec{p}).$ 
For the long-range module, we will give a detailed architecture as follows. A visual representation of the architecture can be found in Appendix \ref{sec:long-range-rep}. The design principles of this long-range module adopt the following proposition:

\begin{proposition}
The Hadamard product of a scalar representation by a vectorial representation results in a vectorial representation.       The inner product of two vectorial representations results in a scalar representation.
\end{proposition}

\subsection{Long-Range Module in ViSNet-LSRM}
\label{sec:vis-lsrm}

\textbf{Layer Normalization of Fragment Representation}: For each fragment, we commence by applying a layer norm to its scalar embedding and norm-based min-max normalization to the  vectorial embedding:
{\small
\begin{equation}
    \vec{V}_{\text{norm}, i} =\frac{\vec{V_i}}{||\vec{V_i}||}\cdot \frac{ ||\vec{V_i}|| - \min(\vec{||V_i||}) }{(\max(\vec{||V_i||}) - \min(\vec{||V_i||})) }.
\end{equation}
}$\min(\cdot)$ and $\max(\cdot)$ are applied to the channel dimension, $||\cdot||$ is $l_2$-norm applied to the spatial dimension. Empirically, we observed that this normalization is a succinct and effective technique to bolster model convergence.

\textbf{Distance-Dependent Bipartite Geometric Transformer}: Considering a central atom, certain fragments may establish robust long-range interactions due to factors such as charge, polarity, distance, and so forth, while others may not. It is intuitive that fragments in close proximity generally exert a more substantial influence than those farther away. In light of these observations, we propose a distance-based equivariant multi-headed attention mechanism for atom-fragment message passing. For simplicity, we will omit the notation for heads in the subsequent derivations. We commence by encoding the distance between fragments and atoms utilizing continuous filters:
\begin{equation}
    \label{eq:attention_edge_rbf}
    s_{ij} = \textsc{Dense}\left(e^{\text{rbf}}(\Tilde{r}_{ij})\right),
\end{equation}
 $e^{\text{rbf}}(\cdot)$ is the radial basis function. The distance between the fragment and the atom depends on the size of both the atom and the fragment, as well as the position of the atom relative to the fragment, we employed a \textbf{size-dependent distance encoding} which is parametrized by a linear function:
\begin{equation}
    \Tilde{r}_{ij} = w(z_i, H_j) \text{dist}(\vec{p_i}, \vec{P_j}) + b(z_i, H_j).
\end{equation}
 
Subsequently, we employed the attention mechanism 
\cite{DBLP:journals/corr/VaswaniSPUJGKP17,tholke2022torchmd} for atom-fragment interaction:
\begin{equation}
    \label{eq:attentionQKV}
    q_i = W_q x_i,\ k_j = W_k H_j,\ v_j = W_v H_j,
\end{equation}
where $W_q$, $W_k$, $W_v$ are projection maps. The query is computed using atom scalar embeddings; the key and value are computed using fragment scalar embeddings.
The attention weights $A_{ij}$ between atom-$i$ and fragment-$j$ are obtained through element-wise multiplication of query $q_i$, key $k_j$, and encoded distance  $s_{ij}$:
\begin{equation}
        A_{ij} = \text{SiLU}(\text{sum}(q_i  \odot k_j\odot s_{ij})),
\end{equation}

The output of the attention mechanism results in a value vector weighted by attention weights, from which a scalar message is derived using a dense layer. Concurrently, the relative direction of the fragment to each atom and the vector embeddings of each fragment form the vectorial message, both gated by a dense layer acting on the attention output:
\begin{equation}
            m_{ij} =  \textsc{Dense}(A_{ij} v_j),
\end{equation}
 {\small
\begin{equation}
        \label{eq:vector_message_from_atom_to_fragments}
  \vec{m}_{ij} = \textsc{Dense}(A_{ij} v_j) \odot \frac{\vec{p}_i - \vec{P}_j }{||\vec{p}_i - \vec{P}_j ||}  + \textsc{Dense}(A_{ij} v_j) \odot\Vec{V}_j.
\end{equation}
}%
The aggregation function of the scalar and the vectorial message is a summation over bipartite fragment neighbors:
\begin{equation}
    m^l_{i} = \sum_{j\in N(i)} m^l_{ij};\ \vec{m}^l_i = \sum_{j\in N(i)} \vec{m}^l_{ij}.
\end{equation}
Define $\hat{U}_i = \left\langle U\left(\vec{\mu}_i^{l-1}\right),U\left(\vec{\mu}_i^{l-1}\right)\right\rangle$, that is the scalar product of vectorial embeddings under two different projections. The updated scalar embeddings and vectorial embedding are obtained as:
\begin{equation}
    x^l_i = x^{l-1}_i +  \textsc{Linear}\left(m^l_{i}\right) \odot \hat{U}_i    + \textsc{Linear}\left(m^l_{i}\right);\ \vec{\mu}^l_i = \vec{\mu}^{l-1}_i +  \Vec{m}_i^l + \textsc{Linear}(m^l_i) \odot \vec{\mu}_i^{l-1}.
\end{equation}

\section{Experiment Results}
\subsection{Results on MD22 Dataset and Chignolin Dataset}
MD22 is a novel dataset presenting molecular dynamics trajectories, encompassing examples from four primary classes of biomolecules: proteins, lipids, carbohydrates, nucleic acids\footnote{BRICS fragmentation methods fail to fragment the supramolecules in MD22, we detail their implementation in LSR-MP in Appendix~\ref{app:graph-cluster} }. The data split of MD22 is chosen to be consistent with~\cite{https://doi.org/10.48550/arxiv.2209.14865}. Chignolin, the most elementary artificial protein, boasts 166 atoms. The dataset from \cite{wang2022ViSNet} features 9,543 conformations, determined at the DFT level. It poses significant challenges given its origin from \textit{replica exchange MD (REMD)}. The REMD method, integrating MD simulation with the Monte Carlo algorithm, adeptly navigates high energy barriers and adequately samples the conformational space of proteins. This results in diverse modes from the conformational space being represented in both the training and test sets. We partitioned the dataset into training, validation, and test sets in an 7:1:2 ratio. For benchmarks, we considered sGDML~\cite{doi:10.1126/sciadv.adf0873}, ViSNet~\cite{wang2022ViSNet}, PaiNN~\cite{schutt2021equivariant}, ET~\cite{tholke2022torchmd}, and the recently introduced GemNetOC~\cite{gasteiger2022gemnet}, MACE~\cite{batatia2022mace}, SO3krates~\cite{frank2022so3krates}, Allegro~\cite{musaelian2023learning} and Equiformer~\cite{liao2022equiformer}. Comprehensive settings and implementation details are provided in Appendix \ref{app-baslines}.

\begin{table*}[ht]
\vspace{-3mm}
\caption{Mean absolute errors (MAE) of energy (kcal/mol) and force (kcal/mol/$\angstrom$) for five large biomolecules on MD22 and Chignolin. The best one in each category is highlighted in bold.}
\vspace{-2mm}
\begin{threeparttable}
{\small
\label{results-table-1}
\resizebox{\linewidth}{!}{
\begin{tabular}{llcl|ccccccccc|>{\columncolor{Gray}}c>{\columncolor{Gray}}c}
\toprule
Molecule                      & Diameter (Å) & \# atoms &        & sGDML  & PaiNN          & TorchMD-NET  &GemNetOC~\cite{shoghi2023molecules}   & SO3krates~\cite{frank2023peptides}  &  Allegro    & MACE~\cite{kovacs2023evaluation} & Equiformer & ViSNet & Equiformer-LSRM & ViSNet-LSRM     \\ \midrule
\multirow{2}{*}{Ac-Ala3-NHMe} & \multirow{2}{*}{$\sim$12} & \multirow{2}{*}{42}  & energy & 0.3902 & 0.1168         & 0.1121 & - &  0.337       & 0.1019        & \textbf{0.0620} & 0.0828 & 0.0796 & 0.0780 & 0.0654 \\
                              &                           &                      & forces & 0.7968 & 0.2302         & 0.1879 & 0.1169 & 0.244      & 0.1068        & 0.0876 & \textbf{0.0804} & 0.0972 & 0.0887 & 0.0902 \\ \midrule
\multirow{2}{*}{DHA}          & \multirow{2}{*}{$\sim$14} & \multirow{2}{*}{56}  & energy & 1.3117 & 0.1151         & 0.1205 & - & 0.379      & 0.1153        & 0.1317 & 0.1788 & 0.1526 & 0.0878 & \textbf{0.0873} \\
                              &                           &                      & forces & 0.7474 & 0.1355         & 0.1209 & 0.0662 &  0.242     & 0.0732        & 0.0646 & \textbf{0.0506} & 0.0668 & 0.0534 & 0.0598 \\ \midrule
\multirow{2}{*}{Stachyose}    & \multirow{2}{*}{$\sim$16} & \multirow{2}{*}{87}  & energy & 4.0497 & 0.1517         & 0.1393 & - & 0.442       & 0.2485        & 0.1244 & 0.1404 & 0.1283 & 0.1252 & \textbf{0.1055} \\
                              &                           &                      & forces & 0.6744 & 0.2329         & 0.1921 & 0.0888 &  0.435      & 0.0971        & 0.0876 & 0.0635 & 0.0869 &\textbf{0.0632} & 0.0767  \\ \midrule
\multirow{2}{*}{AT-AT}        & \multirow{2}{*}{$\sim$22} & \multirow{2}{*}{60}  & energy & 0.7235 & 0.1673         & 0.1120 & - &  0.178      & 0.1428        & 0.1093 & 0.1309 & 0.1688 & 0.1007 & \textbf{0.0772} \\
                              &                           &                      & forces & 0.6911 & 0.2384         & 0.2036 & 0.1241 &   0.216     & 0.0952        & 0.0992 & 0.0960 & 0.1070 & 0.0881 & \textbf{0.0781} \\ \midrule
\multirow{2}{*}{AT-AT-CG-CG}  & \multirow{2}{*}{$\sim$24} & \multirow{2}{*}{118} & energy & 1.3885 & 0.2638         & 0.2072 & - & 0.345       & 0.3933        & 0.1578 & 0.1510 & 0.1995 & 0.1335 & \textbf{0.1135} \\
                              &                           &                      & forces & 0.7028 & 0.3696         & 0.3259 & 0.1296 &  0.332       & 0.1280        & 0.1153 & 0.1252 & 0.1563 & 0.1065 & \textbf{0.1063} \\ \midrule
\multirow{2}{*}{Chignolin}    & \multirow{2}{*}{$\sim$37} & \multirow{2}{*}{166} & energy & -     & 2.4491         & 2.5298  & - &   -    & 3.513         & -     & 1.0967 & 2.4355 &\textbf{0.6687} & 1.2267  \\
                              &                           &                      & forces & -     & 0.6826         & 0.6519  & - &   -    & 0.382         & -     & 0.2121 & 0.3717 &\textbf{0.1867} & 0.2778  \\ \bottomrule

\end{tabular}}
}
\end{threeparttable}
\label{table:main}
\end{table*}
\vspace{-0.1in}
\begin{table}[htbp]
\vspace{-3mm}
\centering
\caption{Comparison of the number of parameters and training speed of various methods when forces MAE is comparable on the molecule AT-AT-CG-CG. Detail settings can be found in Appendix~\ref{app:detail-table}.}
\label{table: number-of-parameters}
\resizebox{\linewidth}{!}{
\begin{tabular}{lccccccc}
\toprule
Methods (MAE)        & ViSNet (0.16) & ViSNet-LSRM (0.13)     & PaiNN (0.35) & ET (0.29) & Allegro (0.13) & Equiformer (0.13)\\ \midrule
\# of Parameters & 2.21M & \textbf{1.70M} & 3.20M & 3.46M & 15.11M & 3.02M\\
Training Time / Epoch (s) & 44 & \textbf{19} & \textbf{19} & 26  & 818 & 155 \\ \bottomrule
\end{tabular}
}
\vspace{-2mm}
\end{table}

As illustrated in Table~\ref{table:main}, the mean absolute errors (MAE) of both energy and forces have significantly decreased in comparison to the original ViSNet, with a more pronounced improvement in performance as the system size expands. This enhancement highlights the importance of explicitly characterizing long-range interactions for large biomolecules. In comparison to previous state-of-the-art models such as Equiformer, Allegro, and sGDML, ViSNet-LSRM demonstrates competitive performance while delivering superior efficiency. Furthermore, we conducted a comparative study to validate the efficacy of our approach in terms of both performance and parameter counts. Our results, presented in Table~\ref{table: number-of-parameters}, show that when performance is comparable, our framework uses fewer parameters and offers faster speed, indicating that our approach is more efficient and effective at capturing long-range interactions. This comprehensive analysis underscores the advantages of ViSNet-LSRM in terms of both performance and efficiency, making it a promising solution for large biomolecular systems.

\begin{figure*}[ht]
\centering
\begin{minipage}{\textwidth}
\includegraphics[width=1\textwidth]{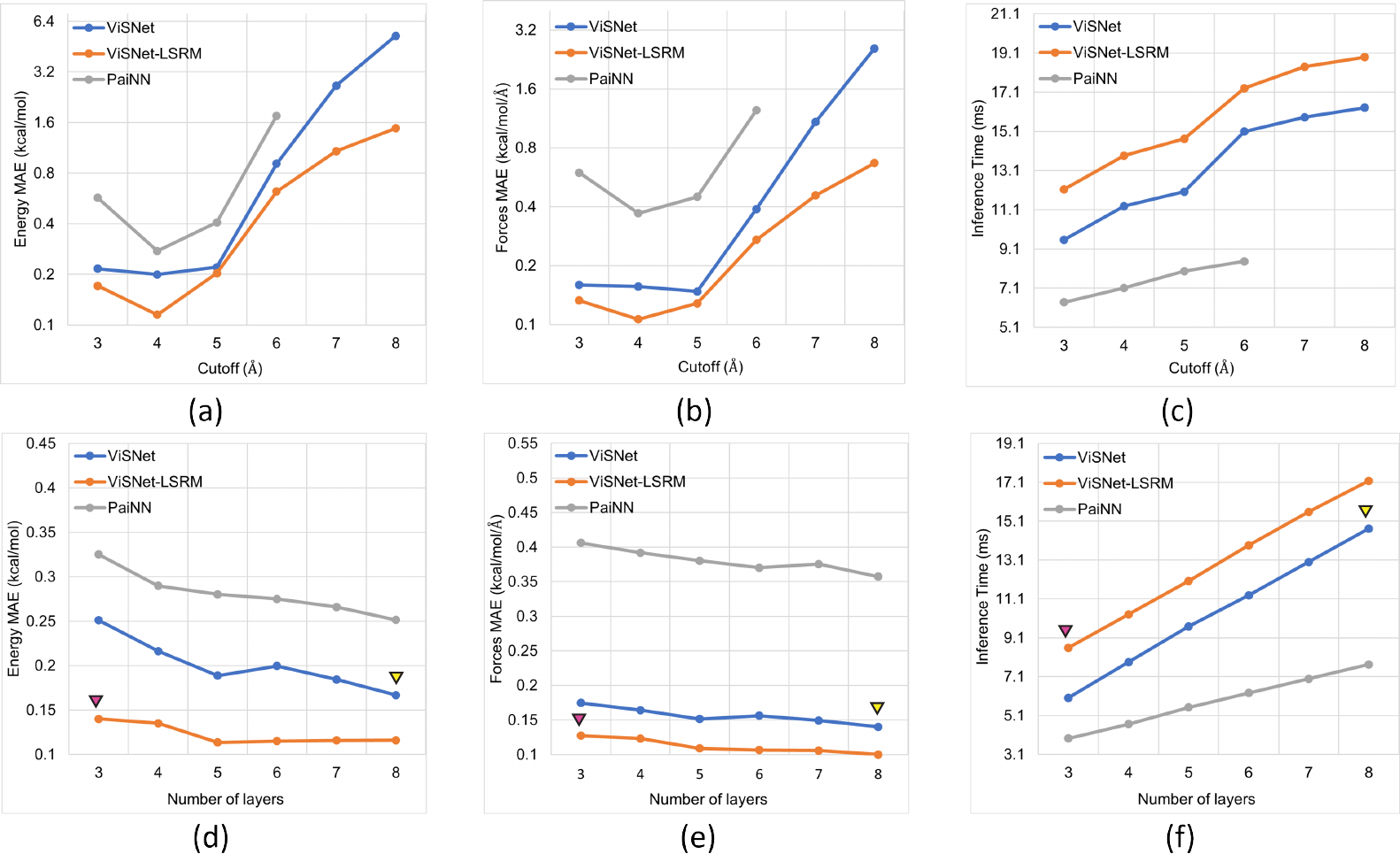}
\end{minipage}

\caption{Comparative Studies on the (short-range) cutoff and the number of (short-range) layers for three methods, including ViSNet-LSRM, ViSNet, and PaiNN. For ViSNet-LSRM, we consistently used 2 long-range layers and a 9 \angstrom \ long-range cutoff. The mean absolute errors (MAE) of energy and forces for \textit{AT-AT-CG-CG} are shown in each sub-graph. From top to bottom, the performance of three methods regarding the cutoff and the number of interactions are shown. a, b, c illustrate increasing the cutoff radius could lead to information over squashing which is prevalent across all EGNNs, including the short-range component of ViSNet-LSRM. d, e, f demonstrate increasing the depth (shown horizontally) of the short-range models is not as effective as introducing the long-range layer (shown vertically).  }
\label{fig:cutoff-interactions-ablation}
\vspace{-0.2in}
\end{figure*}

\subsection{Study of Long-Range Interactions}
\label{sec:study-of-long-range}
In this subsection, we conduct an analysis of the long-range interactions based on the molecule AT-AT-CG-CG in MD22, aiming to address the following questions:

\textbf{Question 1.} Can existing EGNNs capture long-range interactions by increasing the radius cutoffs?

Contrary to expectations, the answer is negative. As depicted in Fig.\ref{fig:cutoff-interactions-ablation}(a) and (b), all three models exhibit optimal performance when the short-range cutoff is set to 4 or 5. A further increase in cutoff could significantly deteriorate the performance. This finding suggests that existing EGNNs may be bottlenecked by their limited capacity in information representation when handling a large number of neighbors, leading to information over-squashing.

\textbf{Question 2.} Are existing EGNNs capable of capturing long-range interactions by increasing the number of layers?

Intuitively, deeper models possess a larger receptive field to capture long-range interactions. The results in Figures \ref{fig:cutoff-interactions-ablation}(d) and (e) across all three models with varying layer numbers confirm this intuition. For ViSNet-LSRM, we vary the number of short-range layers while fixing the long-range layers at 2. Performance improves as layer number increases for all models. However, a 3-layer ViSNet-LSRM (marked in red in Fig. \ref{fig:cutoff-interactions-ablation} outperforms 8-layer ViSNet (marked in yellow in Fig. \ref{fig:cutoff-interactions-ablation} and PaiNN, suggesting incorporating explicit long-range modules is more effective than simply deepening EGNNs. Furthermore, the shallow ViSNet-LSRM is more efficient than the deeper baselines (Fig. \ref{fig:cutoff-interactions-ablation}(f)). \cite{pmlr-v202-di-giovanni23a} has pointed out deepened models would inevitably induce vanishing gradients.

\textbf{Question 3.} How to demonstrate that the proposed method improves the performance due to successful long-range interactions? How to justify the BRICS-fragmentation method?

To address these questions, we conducted an ablation study in Table \ref{tab:ablation-main}, modifying key components of our framework. Instead of the BRICS fragmentation method, we employed distance-based k-means clustering, keeping the rest of the model aligned with LSR-MP. Results indicate that even without chemical insights, the k-means-based LSR-MP surpasses the baseline that lacks LSR-MP, underscoring the importance of incorporating long-range components. However, the k-means method fragments chemical systems non-canonically, potentially masking critical long-range interactions. As previously noted, the BRICS method reduces energy loss from breaking chemical bonds, offering more informative fragmentation for molecular modeling. The improvements of BRICS upon K-means have corroborated these findings.  We further benchmarked LSR-MP against another prevalent framework for long-range interaction, incorporating a single global node to counteract information oversquashing. Our improvements on this framework further accentuate the efficacy and robustness of LSR-MP. 

\textbf{Discussion:} Based on the experiments, we hypothesize a two-fold contribution of BRICS methods:
(1) \textbf{Incorporating Chemical Insights}: Leveraging domain-specific motif matching, BRICS refines representation learning to better capture intrinsic chemical environments. With a foundation in chemical locality and appropriate fragment representation learning, the resulting fragment representations could effectively highlight essential compound properties while omitting superfluous information.
(2) \textbf{Prioritizing Negative Curvature Edges}: This formulation builds upon the \textit{Balanced Forman curvature} introduced in \cite{topping2021understanding}. At its core, the \textit{Balanced Forman curvature} is negative when the edge behaves as a bridge between its adjacent first-order neighborhood, while it is positive when the neighborhood stay connected after removing the edge. The phenomenon of over-squashing correlates with edges exhibiting high negative curvature \cite{topping2021understanding}. A key observation we present (See Appendix~\ref{app:cur}) is that bond-breaking induced by BRICS normally have high negative Balanced Forman curvature in the molecular graph, serving as the bridge or the bottleneck to propagate information across different functional groups. To alleviate this, the proposed long-range module rewires the original graph to alleviate a graph’s strongly negatively curved edges. Furthermore, we have proved that LSR-MP could improve the upper bound of the \textit{Jacobian obstruction} proposed in~\cite{pmlr-v202-di-giovanni23a}. The proof is included in Appendix~\ref{app:oversquashing}.

\begin{table}[h]

\centering
\vspace{-0.15in}
\caption{Ablation studies of LSR-MP on AT-AT-CG-CG.}
\vspace{0.05in}
\label{tab:ablation-main}
\resizebox{0.6\columnwidth}{!}{%
\begin{tabular}{@{}lll@{}}
\toprule
Method (Average Fragment Size)    & Energy MAE                                            & Force MAE                                                      \\ \midrule
BRICS (8.43)                      &  \textbf{0.1064} & \textbf{0.1135} \\
K-Means (8.43)                       & 0.1276 &  0.1246         \\
Single Global Node           & 0.1482  & 0.1696                                                         \\
w/o Long-range Component & 0.1563                                                & 0.1995                                                         \\ \bottomrule
\end{tabular}%
}
\end{table}

\vspace{-0.2in}

\subsection{General Applicability of LSR-MP}

We evaluated the general applicability of LSR-MP on \textit{AT-AT}, \textit{AT-AT-CG-CG}, \textit{Chignolin} by using  PaiNN, ET, Equiformer, and ViSNet as the short-range module.  To ensure rigorous comparative standards, a 6-layer EGNN model was juxtaposed with an LSR-MP model integrating 4 short-range layers and 2 long-range layers. As shown in Fig.~\ref{fig:backbone-ablation}, all models show significant performance improvements under the LSR-MP framework. The force MAE of Equiformer, ViSNet, PaiNN, and ET is reduced by 15.00\%, 24.88\%, 58.8\%, and 55.0\% in \textit{AT-AT-CG-CG}, while similar improvement is shown on energy MAE. 
This suggests that the LSR-MP framework can extend various EGNNs to characterize short-range and long-range interactions, and thus be generally applied to larger molecules.

\begin{figure}[t]
\vspace{-2mm}
\begin{center}
\includegraphics[width=0.95\columnwidth]{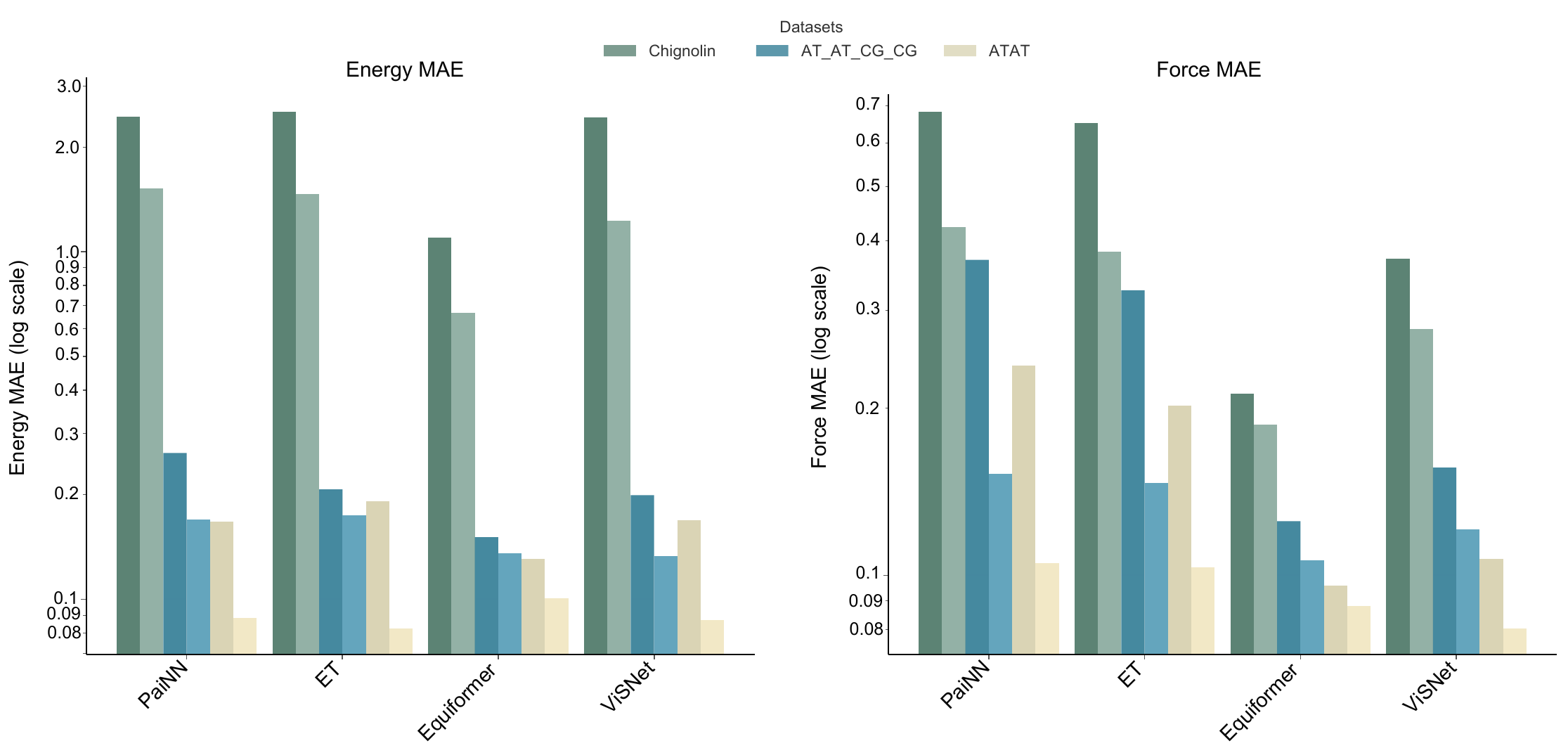}

\caption{General applicability of LSR-MP framework with PaiNN, ET, Equiformer, ViSNet. Dark color schemes indicate an EGNN without LSR-MP, and light color schemes indicate with LSR-MP.}
\label{fig:backbone-ablation}
\end{center}
\vspace{-7mm}
\end{figure} 

\section{Conclusions}
In this work, inspired by the fragmentation-based method in quantum chemistry, we proposed a novel message-passing framework, LSR-MP, to capture long-range and short-range interactions. With an implementation based on ViSNet, it showed superior performance on large molecules with efficiency and effectiveness. It also showed general applicability with application to other EGNNs. 

\textbf{Limitations and Future Work:} As shown in the MD22 dataset, BRICS has failed to deal with supramolecules, and canonical clustering methods were used as alternatives. In our future work, a differentiable fragmentation method may be further developed to learn optimal fragmentation in an end-to-end manner.

\bibliography{iclr2024}
\bibliographystyle{iclr2024_conference}

\newpage

\begin{appendices}
\DoToC

\newpage
\section{Additional Related Work}

\textbf{Learning on Heterogeneous Graphs:} Heterogeneous graphs are vital in recommendation systems \cite{DBLP:journals/corr/abs-1905-04413} and scientific fields like molecular dynamics \cite{wang2022heterogeneous} and systems biology \cite{wang2021kg4sl}. Heterogeneous Graph Neural Networks (HGNNs) are divided into two types: metapath-based and metapath-free methods. Metapath-based HGNNs first gather features from similar semantic neighbors and then combine different semantic types. For example, Relational GCN \cite{schlichtkrull2017modeling} uses unique weights for each relationship type. HAN \cite{wang2019heterogeneous} leverages metapaths to differentiate semantics, merging structural details from each metapath in the aggregation process, and then combining these details for each node. MAGNN \cite{fu2020magnn} expands this idea by considering all nodes in a metapath instance. On the other hand, metapath-free HGNNs, like traditional GNNs, collect information from all neighboring nodes simultaneously. Among metapath-free HGNNs, HetSANN \cite{hong2020attention} uses a multi-layer network to create specific attentions for different relationships, and HGT \cite{hu2020heterogeneous} introduces the Transformer model to handle different types of nodes and edges.

\textbf{Learning Long-Range Dependency:} Recent advancements in graph neural networks (GNNs) have shown promising developments in addressing long-range dependencies within graph structures, a challenge historically difficult due to the inherent limitations of local neighborhood aggregation. Pioneering approaches, like those presented by \cite{li2015gated}, introduced Gated Graph Sequence Neural Networks, leveraging gated recurrent units to better capture long-range interactions. The emergence of Graph Transformers, as proposed by Dwivedi et al. (2020) \cite{dwivedi2020generalization}, marks a significant leap forward, employing self-attention mechanisms to directly model relationships between distant nodes. Additionally, works such as those by \cite{alon2020bottleneck} on spectral-based GNNs have demonstrated innovative ways to encapsulate long-range dependencies by extending traditional convolutional methods. \cite{pmlr-v202-di-giovanni23a} introduced curvature-based rewiring methods to handle long-range dependency. In the realm of equivariant graph neural networks, \cite{frank2022so3krates} introduced extra neighbors in the hidden space to capture long-range interactions. \cite{kosmala2023ewald} introduced Ewald summation, which imposes a Fourier transformation with frequency cutoff to capture long-range interactions.

\section{LSR-MP can alleviate Over-squashing}
\label{app:oversquashing}
The definition of information over squashing is defined as \textit{symmetric Jacobian obstruction} \cite{pmlr-v202-di-giovanni23a}.  The mathematical representation of this concept is given by: \[
        \underbrace{J^m_k(v,u)}_\text{A quantity to characterize oversquashing} := \underbrace{\frac{1}{d_v}\frac{\partial h_v^{m}}{\partial h_v^{k}}}_\text{Sensitivity to myself} - \underbrace{\frac{1}{d_v d_u}\frac{\partial h_v^{m}}{\partial h_u^{k}}}_\text{Sensitivity to distance nodes},
    \]
where $m > k$ both index layer numbers. 
\begin{proposition}[\textbf{EGNN module over-squashing is dependent on commute time} \cite{pmlr-v202-di-giovanni23a}]
    \label{prop:short-range}
    Given a short-range module, with \(S_{r,a} := c_rI + c_a A\) as the graph shift operator adopted by the short-range module. Let \(O^m(v, u)\) be the symmetric Jacobian obstruction of nodes \(v, u\) after \(m\) layers. Assume that each path in the computational graph is activated with equal probability \(\rho\). Let \(\mu\) and \(v\) be the maximal spectral norm and minimal singular value of the weight matrices of the short-range module. Let \(\lambda = \frac{\rho}{\mu c_\alpha 2|E|}\). If \(\mu(c_r + c_a) \leq 1\), we have:
    \[
    \phi(\mathcal{G}_\mathrm{short}, m) \lambda \tau(v,u) \leq O^{(m)}(v,u) \leq \lambda \tau(v,u) \leq 2\lambda |E|^2,
    \]
\end{proposition}
wherein \(\tau(v,u)\) referred to as the \textit{commute time}, quantifies the expected number of steps a random walk\footnote{Random walks refers to a uniform sampling of the adjacency matrix at each step with backtracking enabled.} originating at node 
$v$ would take to reach node $u$ and return. In certain molecular structures, this commute time can be considerably large. Consider, for instance, an elongated carbon chain or unfolded amino acid chain with hundreds of amino acids, the commute time spanning from one end to the other could be $O(|E|^2)$.

 \begin{proposition}[\textbf{Long-range module alleviate over-sqashing}]
 
     Given a long-range module, with the same assumption holds as in Proposition~\ref{prop:short-range}, if \(r_\mathrm{long}\) could cover the whole molecule, the long-range obstruction function is upper bounded by a function independent of \(\mathrm{dist}(\Vec{p_i}, \Vec{p_j})\)
     \[
     O^{(m)} (v,u) \leq \frac{2\rho}{\mu c_\alpha}.
     \]
 \end{proposition}
 \begin{proof}
     The commute time could be expanded as:
     \[
     \tau(v,u) = 2 |E| R(v,u),
     \]
     where \(R(v,u)\) here refers to the voltage difference if a unit current flows from $v$ to $u$, with each edge assigned with unit resistance.

     In long-range module:
     $$\forall v,u \in \mathcal{V}, R(v,u)  = \frac{1}{\frac{1}{R_1} + \frac{1}{R_2} + \cdots \frac{1}{R_n} }, $$
     where $R_1$ $R_2$ refers to the length of the paths from $v$ to $u$, $n$ is the total number of paths.

    We proceed to find a  path from $v,u$ in LSR-MP.
     Recall the fragmentation module introduces virtual nodes connecting to all nodes contained within a fragment. Hence, in the long-range graph, each node is connected to the virtual nodes of its corresponding fragments. Furthermore,  if \(r_\mathrm{long}\) could cover the entire molecule, this translates to a complete bipartite graph where each node is connected to each virtual node. Hence, this gives:
     $$\forall v, u, \exists i\in [n], \text{such that}\ R_i= 2.$$
This gives an upper bound on $R(v,u)$:
$$R(v,u)  = \frac{1}{\frac{1}{R_1} + \frac{1}{R_2} + \cdots \frac{1}{R_n} } \leq R_{\text{min}} \leq  R_i = 2,$$
$$\tau(v,u) \leq 4|E|,$$
     \[
     O^{(m)} (v,u) \leq \frac{\rho \tau(u,v)}{2\mu c_\alpha |E|} \leq \frac{2\rho}{\mu c_\alpha}.
     \]
 \end{proof}

\section{Justification of the Method}
The total energy of a molecule consisting of $N$ atoms writes,
\begin{equation}
    E_\mathrm{tot}^{(N)} = E(\vec{p}_1, \cdots, \vec{p}_N),
\end{equation}
where $\vec{p}_i$ represents the position of the $i$-th atom. Due to the complexity, the exact form of the above function is unknown with limited exceptions of very small systems.
Our development starts from the ansatz that the total energy of the $N$-body system can be decomposed as follows,
\begin{equation}
    E^{(N)}_\mathrm{tot} = \sum_i^N e^{(1)}(\vec{p}_i) + \sum_i^N\sum_{i<j}^N e^{(2)}(\vec{p}_i,\vec{p}_j) + \sum_i^N\sum_{i<j}^N\sum_{i<j<k}^N e^{(3)}(\vec{p}_i,\vec{p}_j,\vec{p}_k) + \cdots,\label{equ:energy-decompose-1}
\end{equation}
where $e^{(1)}$ is the energy associated to a single atom, and $e^{(2)}$ and $e^{(3)}$ are the two- and three-body interactions, respectively. We omit higher orders of interactions in our derivation for simplicity.

To model the total energy function in Eq.~\ref{equ:energy-decompose-1}, existing EGNNs typically adopt the strategy that neglects the interactions beyond a truncation radius $r_\mathrm{short}$, 
\begin{equation}
\begin{aligned}
e^{(2)}(\vec{p}_i,\vec{p}_j) & = 0, \mathrm{if}\ \ \mathrm{dist}(\vec{p}_i,\vec{p}_j) > r_\mathrm{short} \\
e^{(3)}(\vec{p}_i,\vec{p}_j,\vec{p}_k) &= 0, \mathrm{if}\ \ \mathrm{dist}(\vec{p}_i,\vec{p}_j) > r_\mathrm{short}\ \text{or}\ \mathrm{dist}(\vec{p}_j,\vec{p}_k) > r_\mathrm{short}\ \text{or} \ \mathrm{dist}(\vec{p}_k,\vec{p}_i) > r_\mathrm{short},
\end{aligned}
\end{equation}
where $\mathrm{dist}(\cdot,\cdot)$ is a metric function. As a result, Eq.~\ref{equ:energy-decompose-1} is approximated as,
\begin{equation}
\begin{aligned}
    E^{(N)}_\mathrm{tot} & \simeq \sum_i^N e^{(1)}(\vec{p}_i) \\
    & + \sum_i^N\sum_{i<j}^N e^{(2)}(\vec{p}_i,\vec{p}_j)\theta[r_\mathrm{short}-\mathrm{dist}(\vec{p}_i-\vec{p}_j)] \\
    & + \sum_i^N\sum_{i<j}^N\sum_{i<j<k}^N e^{(3)}(\vec{p}_i,\vec{p}_j,\vec{p}_k)\theta[r_\mathrm{short}-\mathrm{dist}(\vec{p}_i-\vec{p}_j)]\theta[r_\mathrm{short}-\mathrm{dist}(\vec{p}_j-\vec{p}_k)]\theta[r_\mathrm{short}-\mathrm{dist}(\vec{p}_k-\vec{p}_i)] \\
    & + \cdots,
\end{aligned}
\label{equ:energy-decompose-2}
\end{equation}
where $\theta[\cdot]$ is the Heaviside step function.

The neglect of the long-range interactions results in loss of accuracy, and detailed discussion can be found in the main text. To accommodate long-range interactions, a na\"ive method is to further increase the truncation radius $r_\mathrm{short}$ in Eq.~\ref{equ:energy-decompose-2}. However, this idea is known to suffer from the loss of data efficiency and prohibitive computational cost as system size increases.

To resolve this issue, we adopt a long-short strategy that introduces a second truncation radius $r_\mathrm{long}$, which is much larger than $r_\mathrm{short}$, and the interactions between atoms are modeled according to their distances relative to $r_\mathrm{short}$ and $r_\mathrm{long}$. In particular,  the two- and three-body interactions can be approximated as follows,
\begin{equation}
\begin{aligned}
e^{(2)}(\vec{p}_i,\vec{p}_j) & \simeq
 \begin{cases}
      \epsilon_2(\vec{p}_i,\vec{p}_j), & \mathrm{dist}(\vec{p}_i,\vec{p}_j) \leq r_\mathrm{short} \\
      \epsilon_2^\prime(\vec{p}_i,\vec{p}_j), & r_\mathrm{short} < \mathrm{dist}(\vec{p}_i,\vec{p}_j) \leq r_\mathrm{long} \\
      0, & \text{otherwise} \\
\end{cases}\\
e^{(3)}(\vec{p}_i,\vec{p}_j,\vec{p}_k) & \simeq 
\begin{cases}
\epsilon_3(\vec{p}_i,\vec{p}_j,\vec{p}_k),& \mathrm{dist}(\vec{p}_i,\vec{p}_j) \leq r_\mathrm{short} \ \text{and}\ \mathrm{dist}(\vec{p}_j,\vec{p}_k) \leq r_\mathrm{short}\ \text{and}\ \mathrm{dist}(\vec{p}_k,\vec{p}_i) \leq r_\mathrm{short}   \\
0,& \mathrm{dist}(\vec{p}_i,\vec{p}_j) > r_\mathrm{long}\ \text{and}\ \mathrm{dist}(\vec{p}_j,\vec{p}_k) > r_\mathrm{long}\ \text{and}\ \mathrm{dist}(\vec{p}_k,\vec{p}_i) > r_\mathrm{long}   \\
\epsilon_3^\prime(\vec{p}_i,\vec{p}_j,\vec{p}_k), & \text{otherwise} \\
\end{cases}\\
& \cdots,
\end{aligned}
\end{equation}
where $\epsilon$ and $\epsilon^\prime$ are approximated forms of the interactions to be learned with neural networks.
Our method is a generalization of existing EGNNs and can be reduced to EGNN in the limit of $r_\mathrm{long}=r_\mathrm{short}$. Apparently, the introduction of second truncation radius $r_\mathrm{long}$ covers significant amount of long-range interactions that do not present in Eq.~\ref{equ:energy-decompose-2}.

Now we turn to the implementation of the method with neural networks. We first assume that the local environment of each atom and the interactions for atoms within the short truncation radius $r_\mathrm{short}$ can be learned through a short-range model,
    \begin{equation}
    \begin{aligned}
        h_i & = \textsc{Short Range Descriptor}(\vec{p}_i) \\
        \epsilon_2(\vec{p}_i,\vec{p}_j) & = \textsc{Short Range Interaction}(h_i,h_j) \\
        \epsilon_3(\vec{p}_i,\vec{p}_j,\vec{p}_k) & = \textsc{Short Range Interaction}(h_i,h_j,h_k) \\
        & \cdots,
    \end{aligned}\label{equ:short-range}
    \end{equation}
where $h_i$ is the descriptor of the local environment for $i$-th atom. The long-range interactions are assumed to be learned also from the local descriptors,
    \begin{equation}
    \begin{aligned}
        \epsilon_2^\prime(\vec{p}_i,\vec{p}_j) & = \textsc{Atomwise Long Range}(h_i,h_j) \\
        \epsilon_3^\prime(\vec{p}_i,\vec{p}_j,\vec{p}_k) & = \textsc{Atomwise Long Range}(h_i,h_j,h_k) \\
        & \cdots.
    \end{aligned}
    \end{equation}
The interactions are written directly in terms of atomic descriptors, thus we name them atomwise long-range models. We note that the atomwise long-range models must not be the same as the short-range ones, as it makes our method an EGNN with a larger cutoff. Without doubt, the computational cost of atomwise long-range model quickly becomes prohibitive with respect to the system size. Thus, we divide the system into multiple fragments, whose descriptors can be learned from atomwise features,
\begin{equation}
    f_\alpha = \textsc{Fragment Descriptor}(\{h_i,\cdots\}),\ i\in F_\alpha,
\end{equation}
and we approximate atomwise long-range interactions with atom-fragment ones:
\begin{equation}
\begin{aligned}
    \epsilon_2^\prime(\vec{p}_i,\vec{p}_j) & + \epsilon_3^\prime(\vec{p}_i,\vec{p}_j,\vec{p}_k) + \cdots  \\
    & = \textsc{Atomwise Long Range}(h_i,h_j) + \textsc{Atomwise Long Range}(h_i,h_j,h_k) + \cdots \\
    & \simeq \textsc{Atom Fragment Long Range}(h_i, \textsc{Fragment Descriptor}(\{h_j,h_k,\cdots\}) ) \\
    & \simeq \textsc{Atom Fragment Long Range}(h_i, f_\alpha ). \\
\end{aligned}
\end{equation}
It is worthy noting that the atom-fragment long-range model only explicitly approximates two-fragment interactions. While long-range fragment-fragment interactions of higher orders are assumed to be negligible, short-range interactions that involve more than two fragments are implicitly encoded in the short-range model (Eq.~\ref{equ:short-range}).

\section{Pitfalls of Fragmentation-based Methods}
Fragmentation-based methods offer scalable solutions for quantum mechanical problems by breaking down large systems into computationally manageable pieces. However, capturing many-body effects remains a challenge. This appendix elucidates the reasons behind these limitations, focusing on the many-body expansion (MBE) as a representative example.

The MBE represents the energy of a system divided into \(N\) fragments as:
\begin{equation}
E_{\text{total}} = \sum_{i} E_i + \sum_{i < j} \Delta E_{ij} + \sum_{i < j < k} \Delta E_{ijk} + \dots
\end{equation}
Where:
\begin{itemize}
    \item \(E_i\) denotes the energy of the \(i\)th fragment computed in isolation.
    \item \(\Delta E_{ij}\) represents the two-body interaction energy between fragments \(i\) and \(j\).
    \item \(\Delta E_{ijk}\) signifies the three-body interaction energy among fragments \(i\), \(j\), and \(k\), and so on.
\end{itemize}

\textbf{Truncation Errors}:
Including higher-order terms often leads to truncation after a certain order (e.g., three-body terms):
\begin{equation}
E_{\text{trunc}} = \sum_{i} E_i + \sum_{i < j} \Delta E_{ij} + \sum_{i < j < k} \Delta E_{ijk}
\end{equation}
For systems where higher-order terms are essential, this truncation fails to depict the many-body nature of interactions.

\textbf{Non-additivity of Many-Body Interactions}: 
Many-body interactions are inherently non-additive, making their accurate representation in fragmentation methods challenging.

\textbf{Basis Set Inconsistency}: 
Basis Set Superposition Error (BSSE) arises due to overlapping basis functions when fragments are combined, distorting the true many-body interactions.

\textbf{Mutual Polarization}: 
Many-body effects often emerge from mutual polarization of multiple fragments. When fragments are treated in isolation or only via pairwise interactions, this mutual polarization is missed.

In conclusiton, while MBE offers a systematic framework for accounting for fragment interactions, its practical implementation in fragmentation methods is fraught with challenges. The inherent approximations, truncation of the series, and non-additive effects introduce difficulties in capturing the true many-body nature of molecular interactions.

\section{Additional Experiments}

\subsection{PubChem}
\label{sec:pubchem}
PubChemQC B3LYP/6-31G*//PM6 database~\cite{nakata2023pubchemqc}, which contains electronic properties calculated using density functional theory for 85.9 million small molecules from PubChem. To testify long-range interactions, we extracted molecules with molecular weight larger than 500 to perform training. The atom number of the dataset ranges from 40  to 100. We used 26000 molecules as the training set and 8000 molecules as the test set. The experiment results are included in Table~\ref{tab:binding-energy}.

\subsection{Electrostatics Binding Energy}

The binding energy of the charged dimers dataset~\cite{Grisafi_2019} consists of 661 diverse organic molecular dimers containing H, C, N, and O atoms, with at least one monomer in each dimer carrying a net charge, extracted from the BioFragment Database. For each dimer, the dataset provides 13 configurations with varying inter-monomer distances from 3-8 Å, and the corresponding binding curves (interaction energies versus distance) are calculated using DFT. 600 dimers are used for training machine learning models and 61 for testing. Isolated monomers are also included to provide the dissociation limit. This dataset offers a realistic challenge for assessing model performance in predicting binding curves dominated by long-range electrostatic interactions across a wide range of chemical environments. The test results are attached in Table~\ref{tab:binding-energy}.

\begin{table}[h]

\centering
\caption{Energy MAE (kcal/mol)}

\label{tab:binding-energy}
\resizebox{0.4\columnwidth}{!}{%
\begin{tabular}{@{}lll@{}}
\toprule
   & Electrostatics Binding Energy & PubChem                                                                                             \\ \midrule
ViSNet                       &  0.1064 &  2.978 \\
ViSNet-LSRM                      &  0.0654 & 2.012   \\ \bottomrule                                                  
\end{tabular}%
}
\end{table}

\subsubsection{Decay of Interactions}
We studied the decay of interactions by separating two molecules in a dimer configuration and plotting the energy as a function of distance.  The results are attached in Fig. \ref{fig:decay}. These experiments demonstrate that, compared to a local model, our model exhibits a more appropriate decaying behavior. This finding is crucial as it suggests that our model captures the long-range interactions more effectively. 

\begin{figure*}[ht]
\centering
\begin{minipage}{\textwidth}
\includegraphics[width=1.0\textwidth]{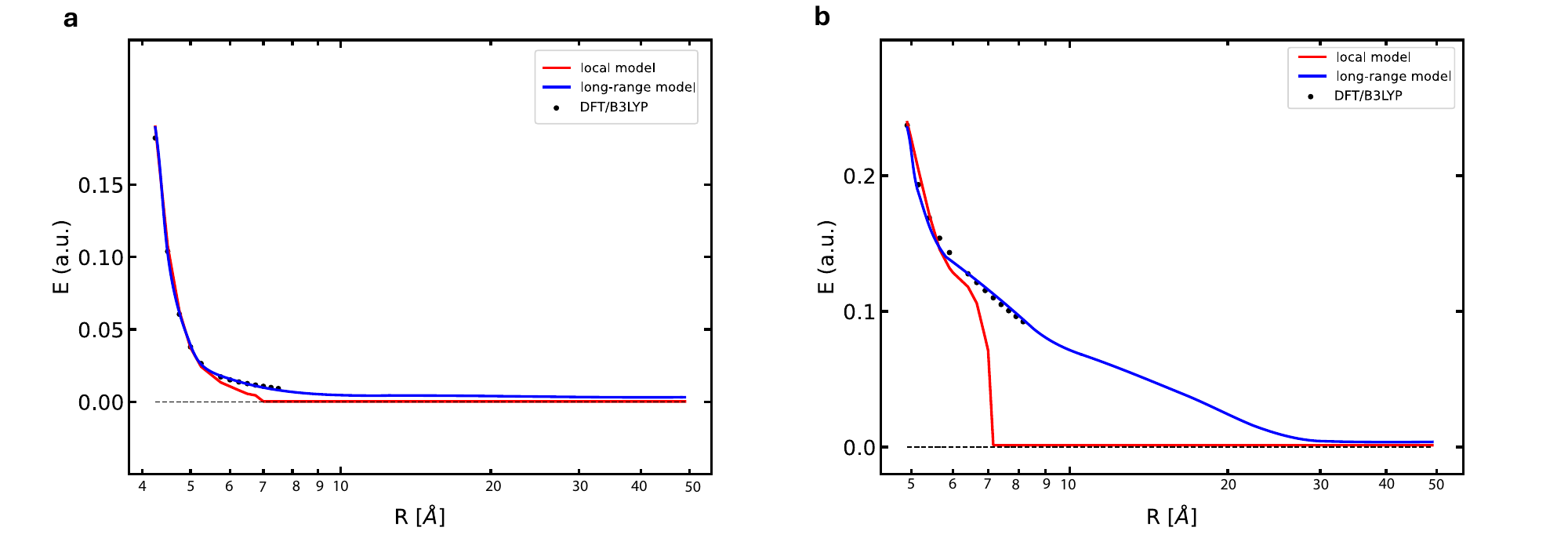}
\end{minipage}
\caption{Visualization of the decay of interaction for electrostatics binding energy dataset. a) The decay curve of one $\text{CH}_3\text{COO}^{-}$ and a 4-Methylimidazole. b) The decay curve of two $\text{CH}_3\text{COO}^{-}$. }
\label{fig:decay}

\end{figure*}

\subsection{Study on Transferability}

In this section, we targeted to examine the transferability or extrapolation capacity of our model to larger unseen molecules. In particular, we performed three experiments to illustrate this:
\begin{itemize}

    \item Zero-Shot Experiment: To study transferability, our first experiment was a zero-shot setup. We trained on molecules including ATAT, Stachyose, DHA, and Ac-Ala3-NhMe, and then tested directly on a larger molecule, ATATCGCG. The zero-shot results are shown in Table \ref{tab:trans}. This experiment revealed that direct transferability without demonstration is challenging for MD22 trajectories.
    \item Few-Shot Learning Experiment: To further explore transferability, we conducted a few-shot learning experiment, as shown in Table \ref{tab:trans}. By adding a small set of 50 ATATCGCG training samples to the original zero-shot training set, our model demonstrated significant improvement over the baseline model. This suggests that with minimal additional training data, our model can adapt to new, larger molecular systems more effectively than traditional local models.
    \item PubChem In our study, we further assessed our model's capabilities using the PubChem dataset, as elaborated in Appendix~\ref{sec:pubchem}.  The dataset features heterogeneous molecules of size ranging from 40 to 100. We recalculated the dataset using t-zvp as the basis set to improve accuracy. Notably, we included molecular force, which remains informative signals given that the molecules were relaxed only through a semi-empirical approach. For dataset division, we used molecules with fewer than 60 atoms (30,545 samples) for training and those with more than 60 atoms (3,455 samples) for testing. Our results in shown in Table \ref{tab:trans}. Compared to the baseline ViSNet model, our model showed enhanced performance on larger molecules, underlining its robust transferability and wide applicability in diverse molecular contexts.
    
\end{itemize}

\begin{table}[ht]
    \centering
    \caption{Transferability Experiments}
    \label{tab:trans}
    \begin{tabular}{@{}clcc@{}} 
        \toprule
        Experiment & Metric & ViSNet & ViSNet-LSRM \\
        \midrule
        \multirow{2}{*}{Zero Shot} & Energy & 184.44 & \textbf{150.23} \\
                                   & Force  &  10.93 &  \textbf{10.21} \\
        \midrule
        \multirow{2}{*}{Few Shot}  & Energy &   2.575 &  \textbf{2.167} \\
                                   & Force  &  0.7448 & \textbf{0.6556} \\
        \midrule
        \multirow{2}{*}{PubChem}   & Energy &   4.458 & \textbf{3.339} \\
                                   & Force  &  0.3303 & \textbf{0.2395} \\
        \bottomrule
    \end{tabular}
\end{table}

\subsection{MD Simulation}
We performed an MD simulation for a relatively large molecule, ATAT, for 20ps, matching the duration of the AT-AT simulation in the MD22 dataset. This was done at a constant energy ensemble (NVE). These simulations were driven by our ViSNet-LSRM and DFT with a time step of $\tau = 1$ fs, allowing us to analyze the vibrational spectra of the AT-AT molecule. As depicted in Fig.~\ref{fig:vel-auto} , both the trajectory in MD22 and the trajectory simulated by ViSNet-LSRM show similar vibrational spectra, albeit with minor differences in peak intensities compared to DFT. This suggests that our simulations can accurately mimic the actual vibrational modes of the molecules over relatively long time periods.

To test the smoothness of ViSNet-LSRM, we ran a longer 200 ps NVE simulation with a time step of $\tau = 1$fs for this molecule. The total energy profile is displayed in Fig.~\ref{fig:smoothness}.  The total energy is conserved within a reasonable range ($+- 0.0001\%$) of fluctuation, validating the capability of the proposed ViSNet-LSRM under long simulations.

\begin{figure*}[ht]
\centering
\begin{minipage}{\textwidth}
\centering
\includegraphics[width=0.8\textwidth]{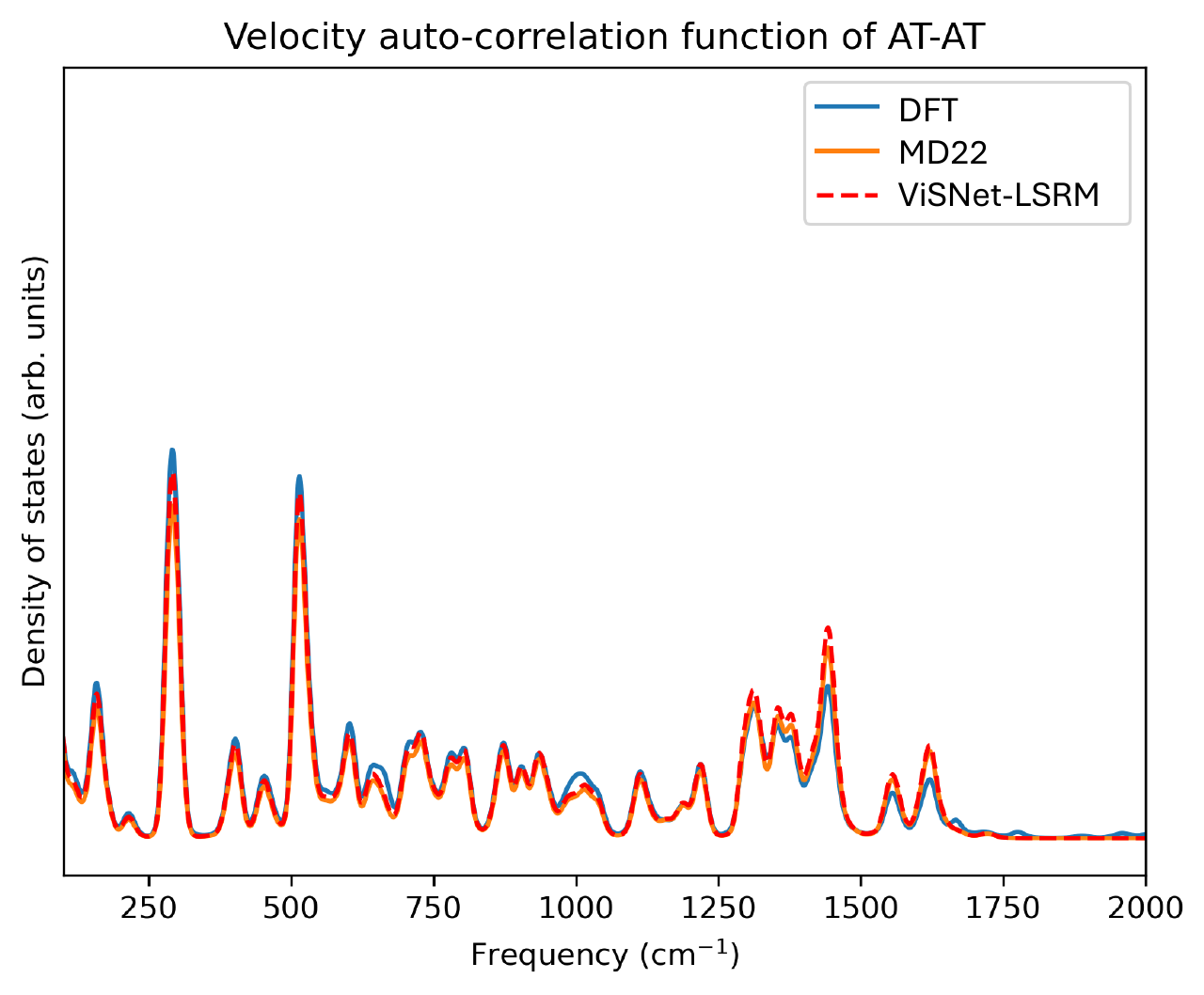}
\end{minipage}
\vspace{-0.15in}
\caption{Velocity autocorrelation function of ATAT.}
\label{fig:vel-auto}

\end{figure*}

\begin{figure*}[ht]
\centering
\begin{minipage}{\textwidth}
\centering
\includegraphics[width=0.8\textwidth]{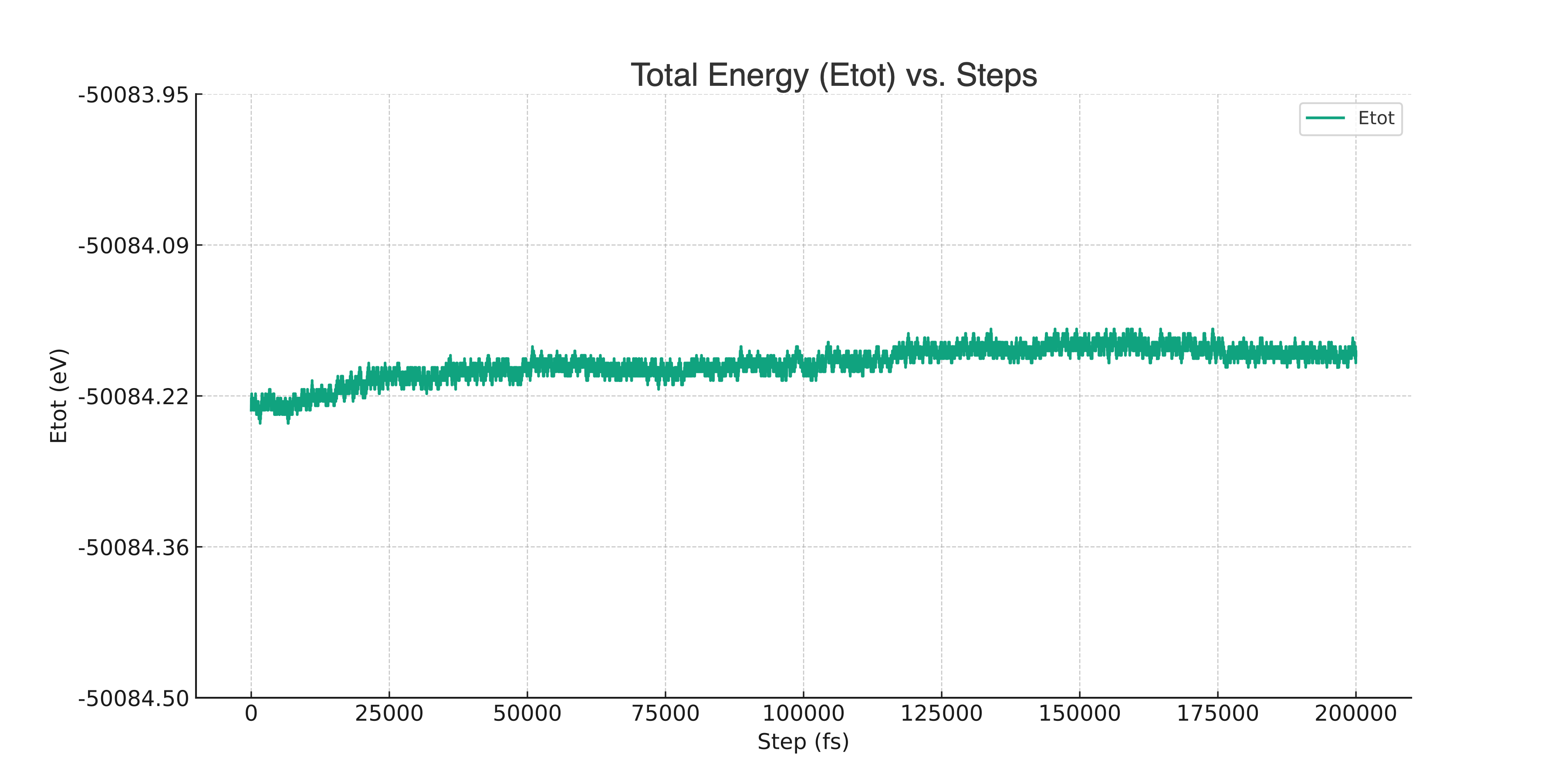}
\end{minipage}
\vspace{-0.15in}
\caption{Visulization of Total Energy. Each step represents 1fs. The fluctuation is within 0.0001\% percentage of the total energy. }
\label{fig:smoothness}

\end{figure*}

\subsection{Examine Curvature of BRICS Prioritized Edges}
\label{app:cur}
To investigate the capacity of BRICS methods to prioritize edges characterized by high negative curvatures, we analyzed the Balanced Forman curvature on two datasets: MD22 and chignolin. In particular, we classified an edge spanning two fragments as a \textit{BRICS-prioritized} edge. Conversely, edges not meeting this criterion were labeled as \textit{Non-BRICS-prioritized} edges. In Figure~\ref{fig:curvature}A, we offer a visualization of the fragmentation outcomes in relation to curvature. Here, the node color scheme represents distinct fragments resulting from BRICS, while the edge coloring reflects the curvature value. This visualization underscores the tendency of BRICS fragmentation to give precedence to edges with negative curvature. To substantiate this observation, we applied the Mann-Whitney U-Test on curvatures across the six systems under scrutiny (See Figure~\ref{fig:curvature}B). Consistently, our findings affirm the propensity of BRICS to prioritize edges with negative curvature.

\begin{figure*}[ht]
\centering
\begin{minipage}{\textwidth}
\includegraphics[width=\textwidth]{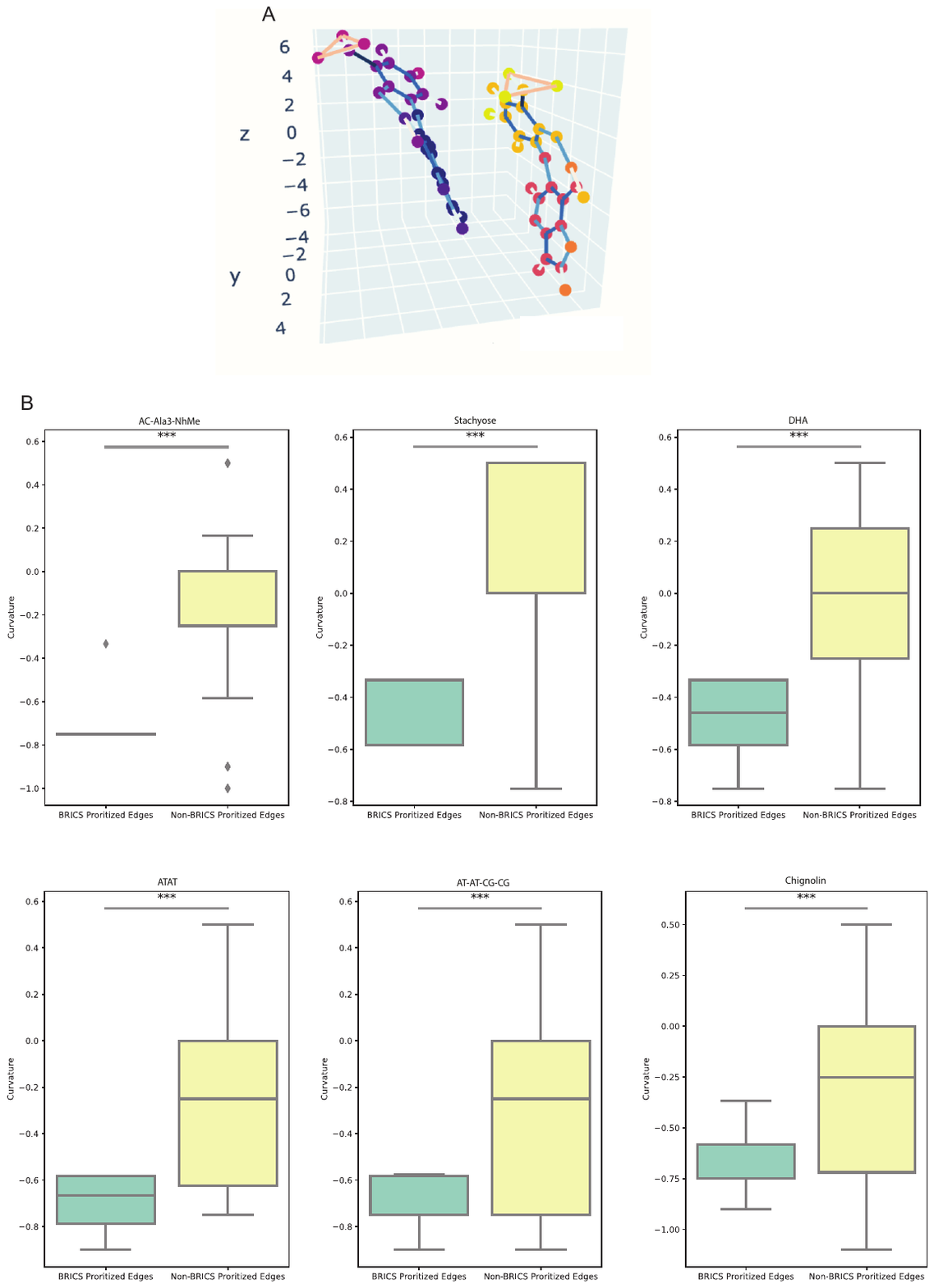}
\end{minipage}
\vspace{-0.15in}
\caption{A. A visualization of BRICS fragmentation and Balanced Forman
curvature. The blue edge indicates negative curvature. The red edge indicates positive curvature. B. Statistical testing of BRICS-prioritized edges and Non-BRICS-prioritized edges. *** means p-value < 0.001 }
\label{fig:curvature}

\end{figure*}

\subsection{Comprehensive Force Evaluation}
We have incorporated additional force evaluation metrics, drawing from the methodologies presented by \cite{wang2023improving}, which encompass:

\textbf{Global Metrics}
\begin{itemize}
\item Mean absolute error (Fmae)
 \item Max absolute error (Fmax-err)
\item Mean normalized error (FNMmae)
 \item Max normalized error (FNMmax-err)
\end{itemize}

\textbf{Element-Based Metrics}:

Fmae, Fmax-err, FNMmae, FNMmax-err for each element type. The results are shown in Figure~\ref{fig:force-eval}

\begin{figure*}[ht]
\centering
\begin{minipage}{\textwidth}
\includegraphics[width=1\textwidth]{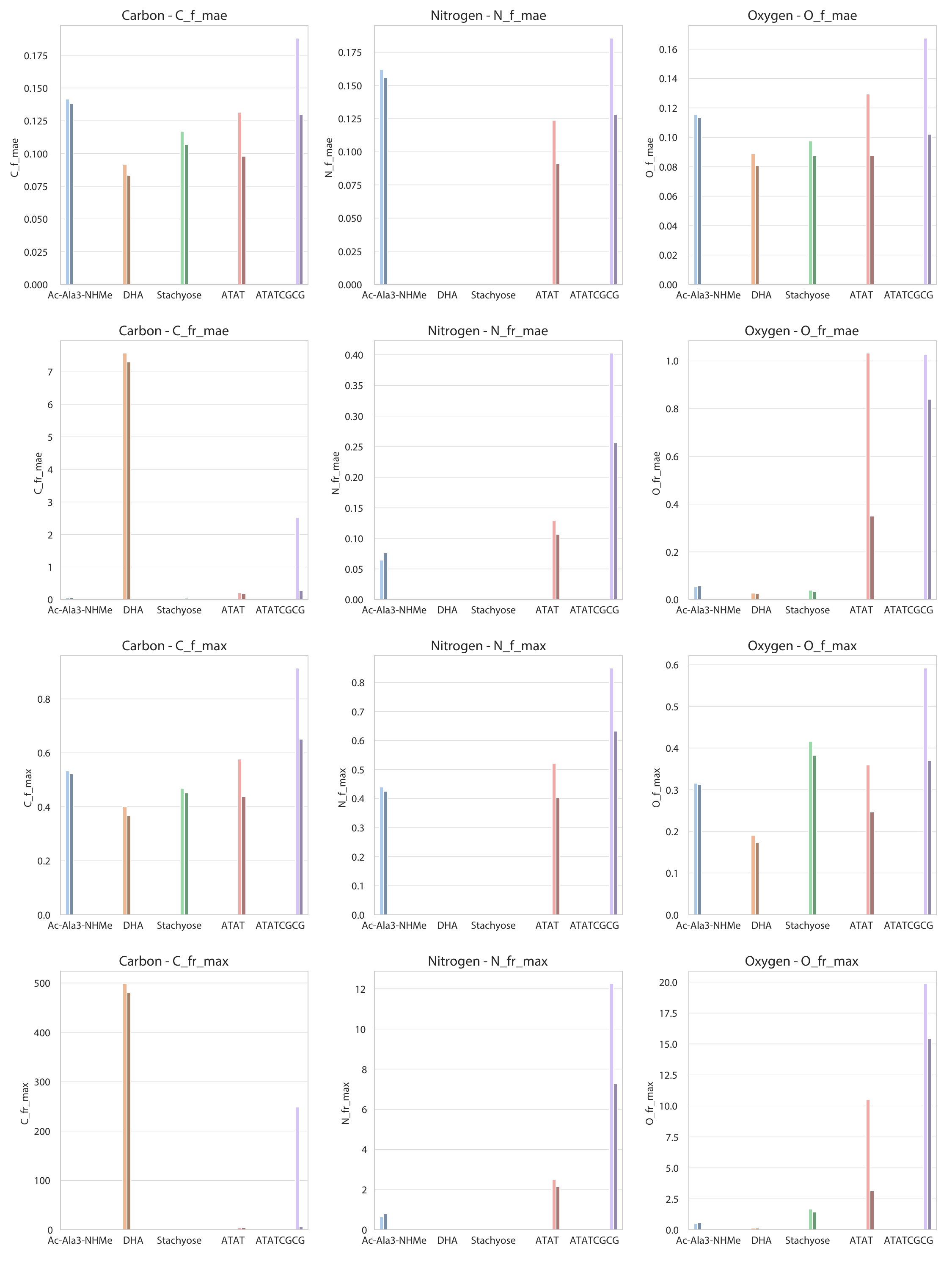}
\end{minipage}
\vspace{-0.15in}
\caption{Element-based force evaluation for ViSNet-LSRM and ViSNet. The dark color indicates ViSNet-LSRM and the light color refers to ViSNet. It could be seen that our models outperform the baseline ViSNet in terms of Mean absolute error (Fmae), Max absolute error (Fmax-err), Mean normalized error (FNMmae), and Max normalized error (FNMmax-err) in most cases. Also, it is worth notifying that ViSNetLSRM has greatly reduced the prediction error on carbon atoms when compared to ViSNet.}
\label{fig:force-eval}

\end{figure*}

\subsection{Comparison of Different Fragmentation Schemes}
\label{app-frag-compare}
The sensitivity of our LSR-MP framework to different fragmentation schemes is a crucial aspect of our research. Fragmentation schemes can be broadly categorized into two groups: knowledge-based fragmentation methods, such as divide-and-conquer (DC) and molecular fractionation with conjugated caps (MFCC), and conventional graph clustering techniques like k-means clustering.

We have conducted a comprehensive analysis of several knowledge-based fragmentation methods, including DC and MFCC, which is shown in Table \ref{tab:frag-compare}. Our findings indicate that the LSR-MP framework is generally robust to various fragmentation schemes. However, some schemes may exhibit superior performance depending on the specific system under investigation.

In addition to knowledge-based methods, we have explored the use of conventional graph clustering techniques by incorporating k-means clustering as an alternative to fragmentation-based methodologies within our LSR-MP framework. Our experiments show that knowledge-based approaches, which draw upon chemical domain expertise, generally outperform k-means clustering methods for most molecules. Nevertheless, the LSR-MP framework, when combined with k-means clustering, consistently surpasses comparable models that do not utilize the LSR-MP approach. Moreover, our results for two supramolecules, employing k-means and distance-based spectral clustering, significantly exceed the performance of equivalent baseline methods.

In conclusion, knowledge-based fragmentation approaches generally outperform k-means clustering methods for the majority of molecules, as k-means is a distance-based clustering method that does not consider chemical properties like atom types and bond types, potentially resulting in chemically insignificant fragments. Furthermore, the LSR-MP framework, when combined with various fragmentation schemes, demonstrates better performance than baseline methods, highlighting the versatility and broad applicability of our method.

\begin{table}[ht]
\centering
\caption{MAE for different Fragmentation schemes on Biomolecules of MD22, the best-performing methods are highlighted in bold.}
\label{tab:frag-compare}
\resizebox{1.0\linewidth}{!}{
\begin{tabular}{@{}llllll@{}}
\toprule
{ \textbf{Molecule}} & { \textbf{Metrics}} & {\textbf{LSRM MFCC}} & {\textbf{LSRM Divide and Conquer}} & {\textbf{LSRM Kmeans}} & {\textbf{LSRM Brics}} \\ \midrule
{ Ac-Ala3-NhMe}      & { Energy}           & {\textbf{0.0637}}     & {0.0824}                            & {0.0662}                & {0.0654}               \\
{ }                  & { Force}            & {\textbf{0.0928}}     & {0.1064}                            & {0.0956}                & {0.0942}               \\ \midrule
{ DHA}               & { Energy}           & {\textbf{0.0815}}     & {0.1374}                            & {0.0966}                & {0.0873}               \\
{ }                  & { Force}            & {\textbf{0.0562}}     & {0.0742}                            & {0.0620}                & {0.0598}               \\ \midrule
{ Stachyose}         & { Energy}           & {0.1295}              & {0.1259}                            & {0.1199}                & {\textbf{0.1055}}     \\
{ }                  & { Force}            & {0.1016}              & {0.0904}                            & {0.0821}                & {\textbf{0.0767}}      \\ \midrule
{ AT-AT}             & { Energy}           & {\textbf{0.0772}}     & {0.1081}                            & {0.1033}                & {\textbf{0.0772}}      \\
{ }                  & { Force}            & {0.0790}              & {0.0929}                            & {0.0911}                & {\textbf{0.0781}}      \\ \midrule
{ AT-AT-CG-CG}       & { Energy}           & {\textbf{0.1135}}     & {0.1438}                            & {0.1446}                & {\textbf{0.1135}}      \\
{ }                  & { Force}            & {\textbf{0.1064}}     & {0.1421}                            & {0.1476}                & {\textbf{0.1064}}      \\ \bottomrule
\end{tabular}}
\end{table}

\subsection{Attention Weight Analysis}
We analyzed the attention coefficients obtained in our model to establish the connection between model predictions and interpretability. 
In particular, we extracted the attention weights in the long-range modules to study the atom-fragment interactions in AT-AT-CG-CG. For each attention head, we visualized the atom-fragment interactions with the largest attention weights, as shown in Figure \ref{fig:att}.  The long-range modules attend to some short-range interactions such as hydrogen bondings (N-NH2) as well as long-range interactions (C-C5H3ON4, C-C5H3N4). This is compatible with the physical intuition that  hydrogen bonds are essential components in the nucleic acids base-pairing system. This also suggests a two-fold contribution of long-range models: (1) they explicitly characterize the long-range interactions and (2) they partly restore  the information lost in short-range message-passing. 
We further studied the attention weights averaged over atoms and fragments, which are shown in Figure \ref{fig:att2} (a,b). We find that the interactions between N-CH3 pair have the smallest attention coefficients, and this can be interpreted as the polarity difference between the nitrogen atom and the methyl group.  We then visualized the two atom-fragment interactions with the largest weights in Figure \ref{fig:att2} (c,d). Figure \ref{fig:att2} (c) suggests the interactions between nitrogenous bases and distant oxygen atoms, counterintuitively, play significant roles in model predictions, which further corroborates our model's capacity to capture interactions beyond local environments.

\begin{figure}[!htbp]
\vskip 0.1in
\centering
\centerline{\includegraphics[width= \textwidth]{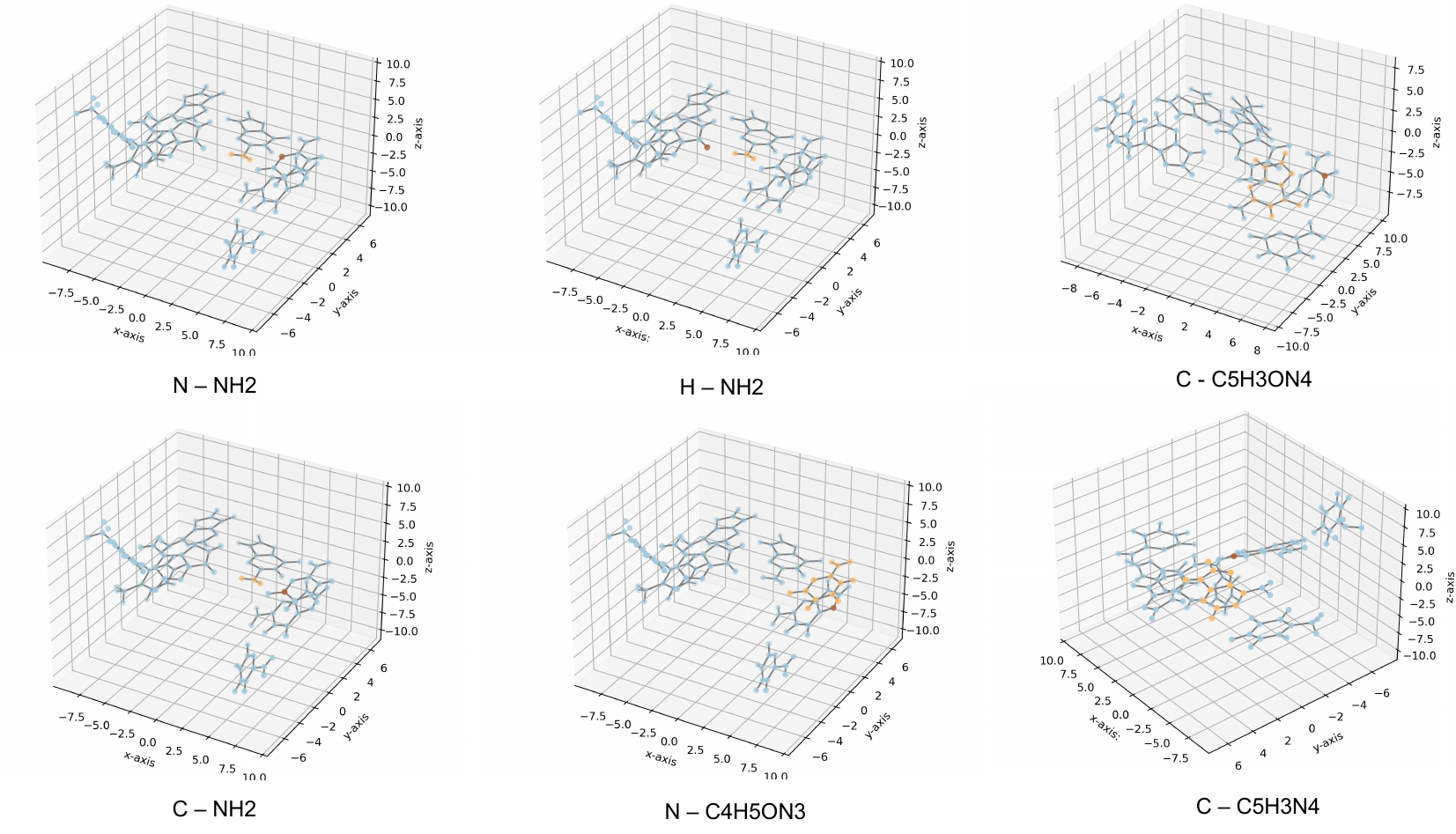}}
\caption{Visulization of atom-fragment interactions with large attention coefficients in AT-AT-CG-CG. Central atoms are denoted in red, and the corresponding fragments are shown in yellow. }
\label{fig:att}
\vskip -0.1in
\end{figure}

\begin{figure}[ht]
\vskip 0.1in
\centering
\centerline{\includegraphics[width= \textwidth]{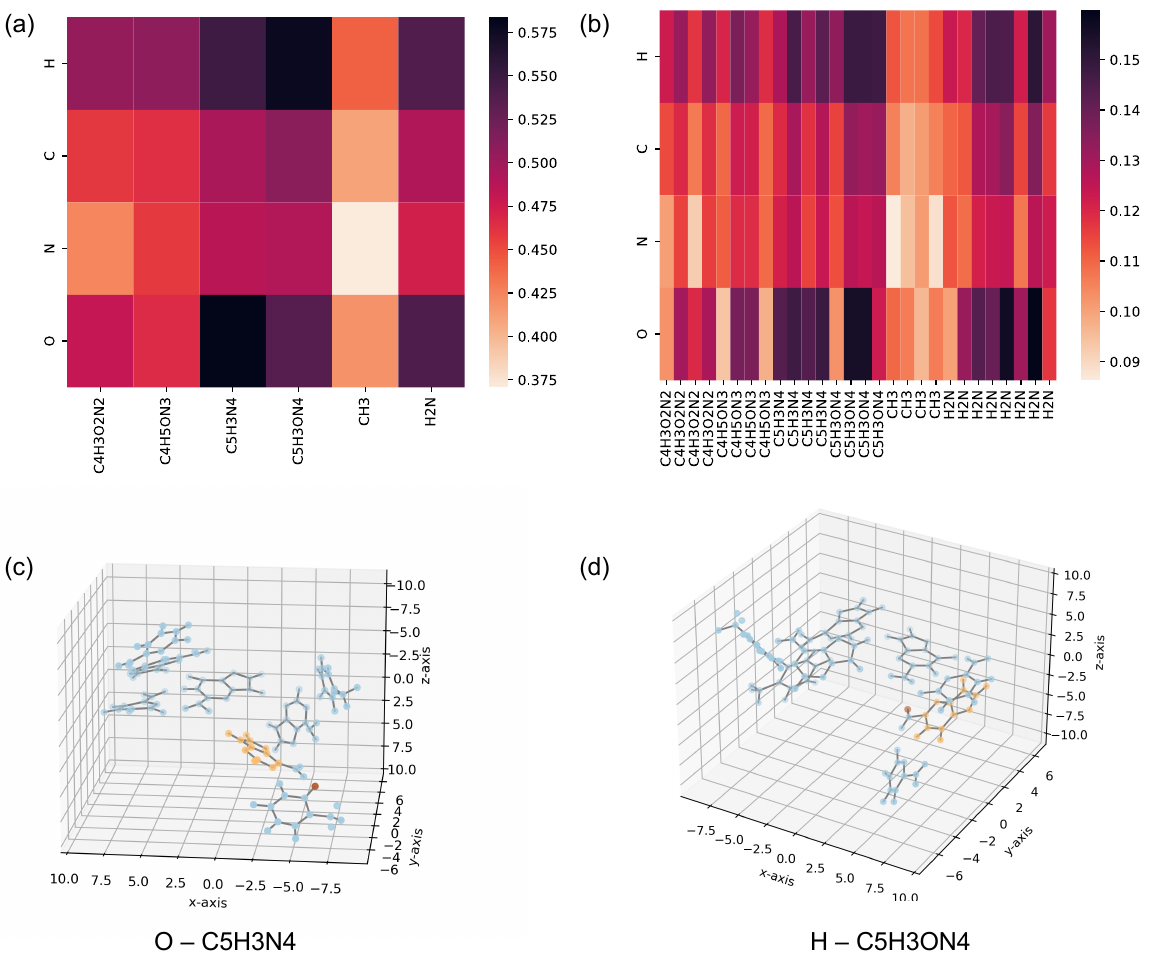}}
\caption{Visulization of attention coefficients in AT-AT-CG-CG. (a) attention coefficients are averaged by atoms and fragments. (b) attention coefficients are only averaged by atoms. (c, d) visualization of the two most salient interactions suggested by (a).}
\label{fig:att2}
\vskip -0.1in
\end{figure}
\subsection{Applicability to Standard Graph Clustering Method}
\label{app:graph-cluster}
For two supramolecules in MD22, \textit{Buckyball catcher} and \textit{Double-walled nanotube}, their conformations do not fall into the chemical prototypes specified in BRICS, leading to a failure in fragmentation. To address this issue, we employ K-Means clustering for \textit{Buckyball catcher} and distance-based spectral clustering for \textit{Double-walled nanotube}.
As demonstrated in Table~\ref{results-table-2}, ViSNet-LSRM achieves competitive performance compared to other EGNN methods for these two supramolecules. This result indicates that our framework is compatible with standard graph clustering methods, which could make the method more universally applicable. Nevertheless, the development of a general fragmentation algorithm for such supramolecules warrants further investigation. In addition, we have conducted a comparative analysis of fragmentation schemes focusing on the BRICS method and standard graph clustering applied to biomolecules in MD22. The results are included in Appendix \ref{app-frag-compare}.

\begin{table*}[ht]
\vspace{-3mm}
\caption{Mean absolute errors (MAE) of energy (kcal/mol) and force (kcal/mol/$\angstrom$) for two supramolecules on MD22 compared with state-of-the-art models. The best one in each category is highlighted in bold.}
\begin{threeparttable}
{\small
\label{results-table-2}
\resizebox{\linewidth}{!}{
\begin{tabular}{lclcccccccc}
\toprule
Molecule                                & \# atoms             &        & sGDML           & PaiNN   & TorchMD-NET   & Allegro       & Equiformer      & ViSNet & ViSNet-LSRM     \\ \midrule
\multirow{2}{*}{Buckyball catcher}      & \multirow{2}{*}{148} & energy & 1.1962          & 0.4563  & 0.5188        & 0.5258        & \textbf{0.3978}          & 0.4421 & 0.4220 \\
                                        &                      & forces & 0.6820          & 0.4098  & 0.3318	    & \textbf{0.0887}        & 0.1114          & 0.1335 & 0.1026 \\ \midrule
\multirow{2}{*}{Double-walled nanotube} & \multirow{2}{*}{370} & energy & 4.0122          & 1.1124  & 1.4732        & 2.2097        & 1.1945          & \textbf{1.0339} & 1.8230 \\
                                        &                      & forces & 0.5231 & 0.9168  & 1.0031        & 0.3428       & \textbf{0.2747}          & 0.3959 & 0.3391 \\ \bottomrule
\end{tabular}}
}
\end{threeparttable}
\end{table*}
\subsection{Ablation study of the long-range cutoff}

As shown in Table.~\ref{ablation-long-range}, the performance of ViSNet-LSRM improves when the long-range cutoff changes from 6 to 9 and fluctuates slightly as it continues to increase. This is likely because all relevant fragments have already been included within 9$\angstrom$, and further increase does not introduce extra information. 
In addition, a large long-range cutoff does not significantly increase the computational cost or lead to the information over-squashing, since the number of fragments is small.
When dealing with larger molecules, increasing the long-range cutoff may be useful and still efficient.
When compared with original ViSNet with 5 layers, ViSNet-LSRM with 3 layers has similar inference time and better performance. 
In conclusion, all studies suggest that our LSR-MP framework is extremely efficient and effective for modeling the long-range interactions rather than either deepening the model or increasing the short-range cutoff.

\begin{table}[!htbp]

\caption{Study of long-range cutoff in LSR-MP framework on \textit{AT-AT-CG-CG} in MD22 dataset. The best results are shown in bold.}
\centering
\vspace{.1in}
\label{ablation-long-range}
\begin{threeparttable}
\resizebox{0.6\linewidth}{!}{
\begin{tabular}{cccc}
\toprule
Long-range cutoff& Energy MAE & Forces MAE & Inference time\\
($\angstrom$) & (kcal/mol) & (kcal/mol/$\angstrom$) & (ms) \\
\midrule
6                 & 0.1234    & 0.1100     & \textbf{12.02} \\
9                 & 0.1135     & \textbf{0.1064}     & 12.14 \\
12                & \textbf{0.1117}     & 0.1074     & 12.26 \\
15                & 0.1166     & 0.1116     & 12.26 \\
\bottomrule
\end{tabular}}
\end{threeparttable}
\vspace{-.1in}
\end{table}

\subsection{Impact of Fragment Size}
In Figure~\ref{fig:fragsize}, we investigate the relationship between the Average Fragment Size on AT-AT-CG-CG and two evaluation metrics: Force MAE and Energy MAE. As the fragment size increase, both metrics exhibit a first decrease and then increase trajectory. Notably, both the Force MAE and Energy MAE share a similar trend, emphasizing the significance of fragment size in influencing these outcomes. This observation underscores the significance of choosing an optimal fragment size.

\begin{figure}[!htbp]
\vskip 0.1in
\centering
\centerline{\includegraphics[width= \textwidth]{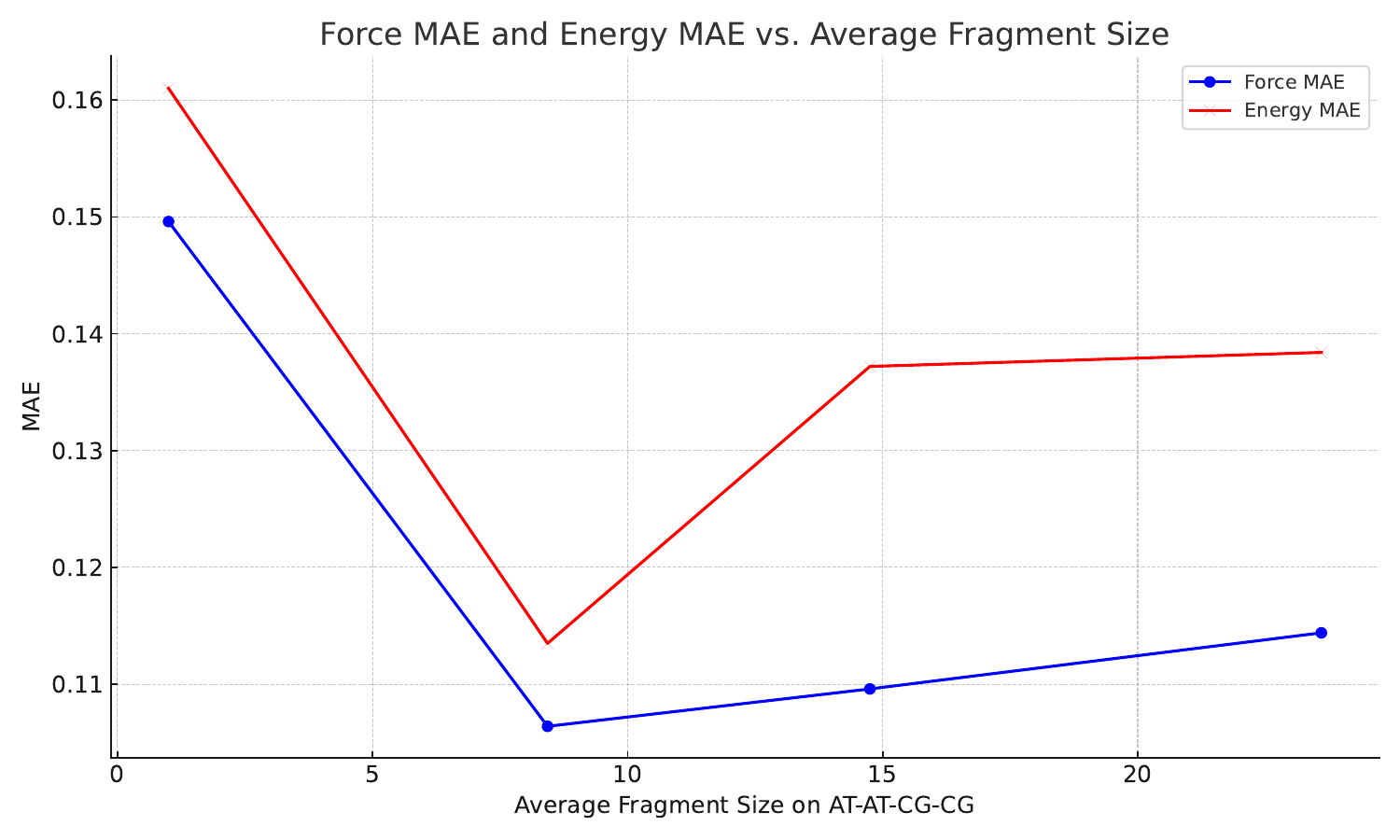}}
\caption{Visulization of  the relationship between the Average Fragment Size on AT-AT-CG-CG and  Model Errors }
\label{fig:fragsize}
\vskip -0.1in
\end{figure}

\section{Geometric Optimization}
\subsection{Geometric Optimization Acceleration with Hybrid Models}
Geometric optimization, often referred to as energy minimization or molecular mechanics optimization, is a computational technique used in molecular modeling. The primary goal of geometric optimization is to find the stable or "lowest-energy" structure of a molecule. This is accomplished by iteratively adjusting the atomic positions to minimize the potential energy of the system, usually using gradient descent or related algorithms. Once the molecule reaches a state where the force on each atom is essentially zero, the molecule is said to be in its optimized geometry or at a local energy minimum. In our evaluation, we randomly sampled 5 initial configurations from the MD22 AT-AT-CG-CG molecule's test set. Using m06-2x DFT, we defined a `reference geometry' through geometric optimization. To assess the neural potential model, we adopted a two-stage approach: initially leveraging the neural potential for optimization, and then refining with DFT until convergence. Our key metrics are the DFT iteration counts required for refinement and the root-mean-square deviation (RMSD) between the neural potential-converged geometry and the reference. Fewer DFT iterations suggest the model's practical utility in GO acceleration. Concurrently, the RMSD provides a direct measure of the model's reliability in mirroring DFT PES. As depicted in Figure~\ref{fig:GO}, ViSNet-LSRM, in comparison to ViSNet, more accurately replicates the DFT-converged geometry. Moreover, ViSNet-LSRM substantially diminishes the number of DFT steps required in hybrid models when compared to ViSNet, achieving acceleration rates of up to 50\%.

\begin{figure}[htbp]
\vspace{-1mm}
\begin{center}
\includegraphics[width=0.9\columnwidth]{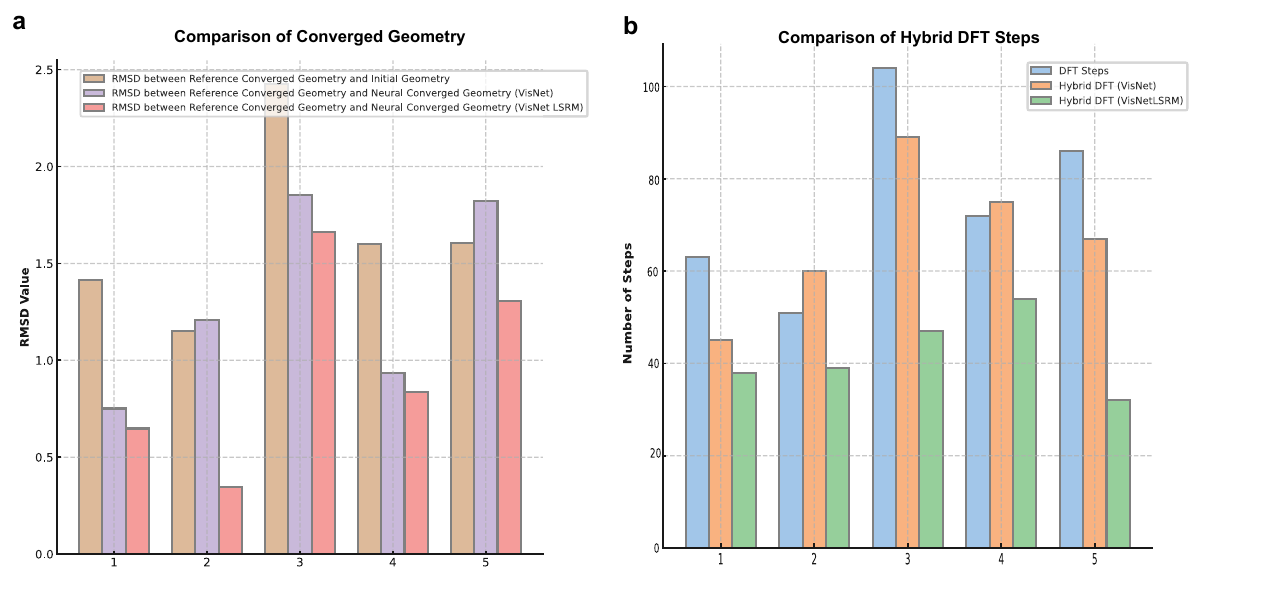}
\caption{Geometric Optimization. x-axis represents five different initial configurations. (a) Comparison of Neural Optimized Geometry when juxtaposed with the DFT optimized Geometry (b) Comparison of the number of DFT steps required in Hybrid Models.}
\label{fig:GO}
\end{center}
\vspace{-5mm}
\end{figure} 

Geometric Optimization (GO) and Molecular Dynamics (MD) are pivotal computational methodologies designed to probe molecular systems and kinetics. While GO primarily scouts the energy landscape for minimal energy configurations, MD provides insights into the temporal changes of molecular structures, accommodating specific thermodynamic ensembles like NVT and NPT. The core objective of both methodologies is to elucidate the dynamic attributes of molecular systems, positioning them as critical tools for understanding kinetics in computational studies.

Our study leveraged Density Functional Theory (DFT) not as a supervised learning component but as an evaluation measure for optimized geometry. We utilized force fields trained on the MD22 dataset, inferring kinetics from potential energy surface gradients. The GO process integrated approximately 100 trajectories, all distinct from the training set, emphasizing the task's relevance in showcasing the force field's generalization and kinetic representation.

Considering five structures in our analysis, each GO trajectory required about 20 minutes on a 24-core CPU, culminating in a week of computation on a 128-core CPU using the PySCF software for DFT tasks. Although computationally intensive, we deem this effort critical for robust insights. We also ensured that the initial configurations for GO were excluded from the Machine Learning Force Field (MLFF) training to avoid overlaps.

In essence, GO operates as a time-series mechanism, transitioning from a specific conformation $s_1$ to $s_t$
  until convergence. Importantly, each step in this process is based on the conformational position of the previous step and the corresponding forces. This is fundamentally similar to the working principle of MD (Molecular Dynamics), whereas MD run under specific thermodynamic ensemble conditions, such as NVT (constant temperature and volume) or NPT (constant temperature and pressure). For better understanding, we enclose a pseudocode for GO as follows:

\begin{algorithm}
\caption{Geometric Optimization with Conjugate Gradient}
\KwData{$molecule$, $forceField$, $energyTolerance$, $gradientTolerance$, $maxGradientTolerance$, $maxIterations$}
\KwResult{$currentStructure$, $trajectoryList$}
\Begin{
    $\text{trajectoryList} \gets \text{Initialize}$\;
    $\text{currentStructure} \gets molecule.\text{getInitialStructure()}$\;
    $\text{currentEnergy} \gets forceField.\text{computeEnergy(currentStructure)}$\;
    $\text{currentGradient} \gets forceField.\text{computeGradient(currentStructure)}$\;
    $\text{searchDirection} \gets -\text{currentGradient}$\;
    $\text{iteration} \gets 0$\;
    $\text{converged} \gets \text{False}$\;
    \While{not converged and iteration $<$ maxIterations}{
        $\alpha \gets \text{LineSearch(currentStructure, searchDirection, forceField)}$\;
        $\text{newStructure} \gets \text{currentStructure} + \alpha \times \text{searchDirection}$\;
        $\text{newEnergy} \gets forceField.\text{computeEnergy(newStructure)}$\;
        $\text{newGradient} \gets forceField.\text{computeGradient(newStructure)}$\;
        $\text{gradientRMS} \gets \sqrt{\text{mean(newGradient}^2)}$\;
        $\text{gradientMax} \gets \text{max(abs(newGradient))}$\;
        $\beta \gets (\text{newGradient} \cdot \text{newGradient}) / (\text{currentGradient} \cdot \text{currentGradient})$\;
        $\text{searchDirection} \gets -\text{newGradient} + \beta \times \text{searchDirection}$\;
        $\text{energyDifference} \gets \text{abs(newEnergy - currentEnergy)}$\;
        \If{$(\text{energyDifference} < \text{energyTolerance})$ or $(\text{gradientRMS} < \text{gradientTolerance})$ or $(\text{gradientMax} < \text{maxGradientTolerance})$}{
            $\text{converged} \gets \text{True}$\;
        }
        \Else{
            $\text{currentStructure} \gets \text{newStructure}$\;
            $\text{currentEnergy} \gets \text{newEnergy}$\;
            $\text{currentGradient} \gets \text{newGradient}$\;
            $\text{trajectoryList.append(currentStructure)}$\;
            $\text{iteration} += 1$\;
        }
    }
    \Return $\text{currentStructure, trajectoryList}$\;
}
\end{algorithm}

For each initial structure (not included in the training set), we commence by performing \texttt{GeometricOptimizationCG} until convergence using DFT FF/ViSNet FF/ViSNet-LSRM FF. Notice that ViSNet FF and ViSNet LSRM FF were trained on MD22 DFT labels.

We extract the last item in the \texttt{trajectoryList}, which we refer to as the converged geometry.
We proceed by comparing the rmsd of the DFT Converged Geometry and initial geometry, ViSNet Converged Geometry and the DFT Converged Geometry, ViSNet-LSRM Converged Geometry and the DFT Converged Geometry.

\section{BRICS and Modification of BRICS}

\label{app-brics}
The BRICS method is a fragmentation technique designed to identify local chemical environments indicated by link atoms of different types. By breaking active/weak bonds, a series of small, active fragments is produced. BRICS takes the chemical environment of each bond type and the surrounding substructures into consideration, resulting in fragment assignments that are more in line with chemistry and reducing the energy loss caused by bond breaking. However, this method can produce too-small fragments, even just one or two atoms. To address this issue, a minimum fragment size and maximum fragment size are set, and any fragment smaller than the minimum is merged with the smallest neighboring fragment if their sum is less than the maximum. This greatly reduces the number of small fragments. A visual representation of BRICS fragmentation results is shown in Figure \ref{fig-brics}.

 The Breaking of Retrosynthetically Interesting Chemical Substructures (BRICS) method is one of the most widely employed strategies in the communities of quantum chemistry, chemical retrosynthesis, and drug discovery. We summarize the key points of BRICS as follows:
\begin{enumerate}
    \item A compound is first dissected into multiple substructures at predefined 16 types of bonds that are selected by organic chemists. In addition, BRICS also takes into account the chemical environment near the bonds, e.g. the types of atoms, to make sure that the size of each fragment is reasonable and the characteristics of the compounds are  kept as much as possible.
    \item BRICS method then applies substructure filters to remove extremely small fragments (for example a single atoms), duplicate fragments, and fragments with overlaps.
    \item Finally, BRICS concludes the fragmentation procedure by adding supplementary atoms (mostly hydrogen atoms) to the fragments at the bond-breaking points and makes them chemically stable. 
\end{enumerate}

\begin{figure*}[htbp]
\vskip 0.1in
\begin{center}
\centerline{\includegraphics[width=1 \columnwidth]{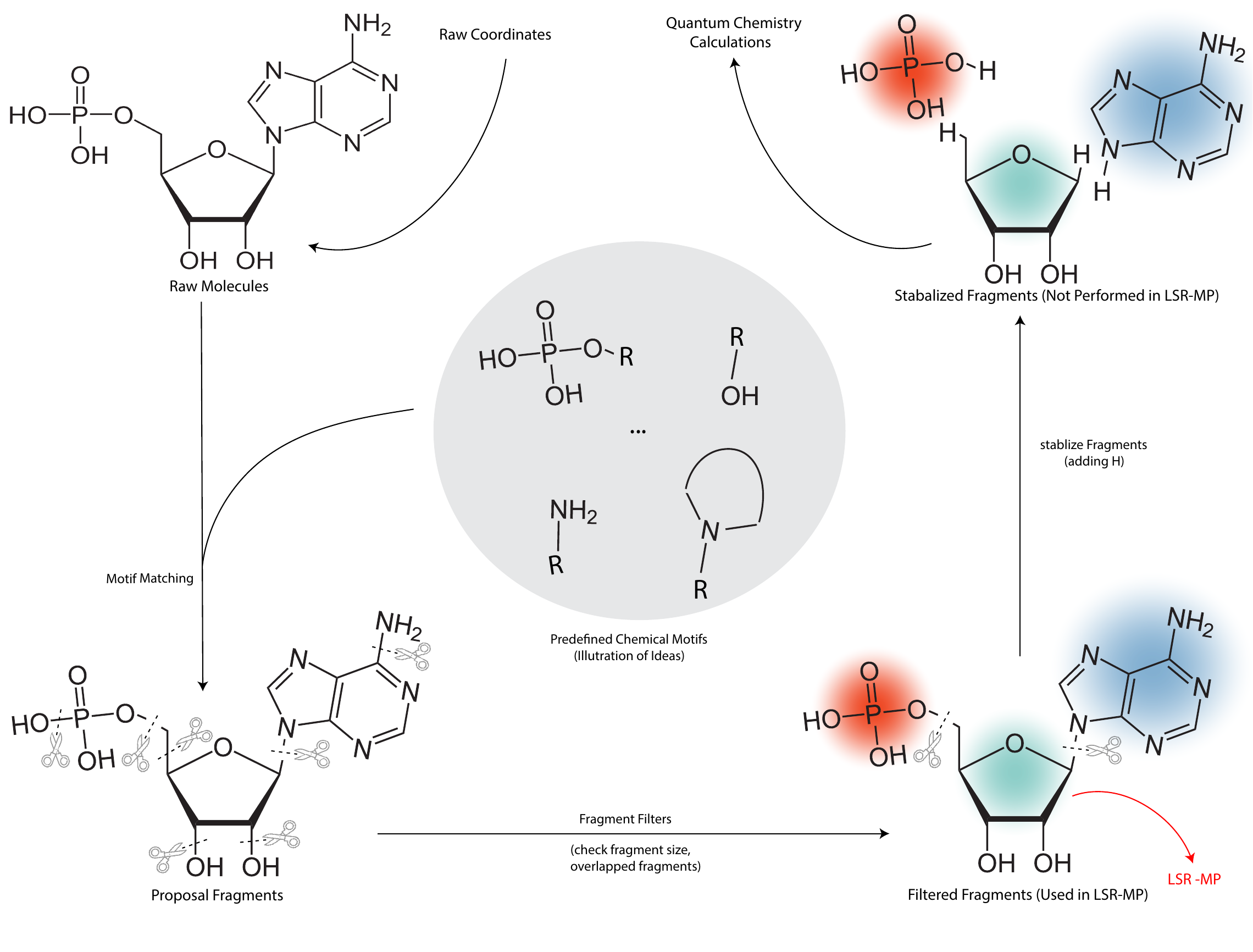}}
\vspace{-0.12in}
\caption{Illustration of the BRICS algorithm and its integration with LSR-MP. }
\label{icml-historical-brics}
\end{center}
\vskip -0.1in
\end{figure*}

\begin{algorithm}

\DontPrintSemicolon
\SetAlgoLined

\SetKwFunction{FMain}{BRICS\_Algorithm}
\SetKwFunction{FBonds}{Find\_Bonds}
\SetKwFunction{FApplyFilters}{Apply\_Substructure\_Filters}
\SetKwFunction{FStabilizeFragments}{Stabilize\_Fragments}
\SetKwFunction{FBreakBond}{Break\_Bond}

\SetKwInOut{Input}{Input}
\SetKwInOut{Output}{Output}

\caption{BRICS Algorithm for Fragmentation}
\label{alg-brics}
\Input{Molecule}
\Output{Set of final fragments}

\FMain{molecule}{
    \Indp\;
    bonds\_to\_break $\leftarrow$ \FBonds{molecule}\;
    fragments $\leftarrow$ \FBreakBond{predefined\_bonds}\;
    filtered\_fragments $\leftarrow$ \FApplyFilters{fragments}\;
    stabilized\_fragments $\leftarrow$ \FStabilizeFragments{filtered\_fragments}\;
    \Return{stabilized\_fragments}\;
    \Indm
}\;

\FBonds{molecule}{
    \Indp\;
    \KwData{molecule}
    \KwResult{List of bonds to break}
    bonds\_to\_break $\leftarrow$ empty list\;
    \For{each bond in the molecule}{
        \If{the bond and its chemical environment match one of the 16 predefined bond types}{
            Add it to the list of bonds to break\;
        }
    }
    \Return{bonds\_to\_break}\;
    \Indm
}\;

\FBreakBond{bonds\_to\_break}{
    \Indp\;
    \KwData{bonds\_to\_break}
    \KwResult{fragments}
    fragments $\leftarrow$ empty list\;
    \For{each bond in the bonds\_to\_break}{break the bond and add resulting fragment to the fragment list
    }
    \Return{fragments}\;
    \Indm
}\;

\FApplyFilters{fragments}{
    \Indp\;
    \KwData{list of fragments}
    \KwResult{Filtered list of fragments}
    filtered\_fragments $\leftarrow$ empty list\;
    \For{each fragment in fragments}{
        \If{fragment size is reasonable and not a duplicate or overlapping with other fragments}{
            Add it to the list of filtered fragments\;
        }
    }
    \Return{filtered\_fragments}\;
    \Indm
}\;

\FStabilizeFragments{fragments}{
    \Indp\;
    \KwData{list of fragments}
    \KwResult{List of stabilized fragments}
    stabilized\_fragments $\leftarrow$ empty list\;
    \For{each fragment in filtered\_fragments}{
        Add supplementary atoms (e.g., hydrogen atoms) to make the fragment chemically stable\;
        Add the stabilized fragment to the list of stabilized fragments\;
    }
    \Return{stabilized\_fragments}\;
    \Indm
}\;

\end{algorithm}

Additionally, to enhance the method's accessibility to a broader audience, we provide a pseudo-code for BRICS in Algorithm~\ref{alg-brics}. For a deeper understanding of the method, we encourage readers to consult the original BRICS paper for further details.
\begin{figure}[ht]
\vskip 0.1in
\centering

\centerline{\includegraphics[width= \textwidth]{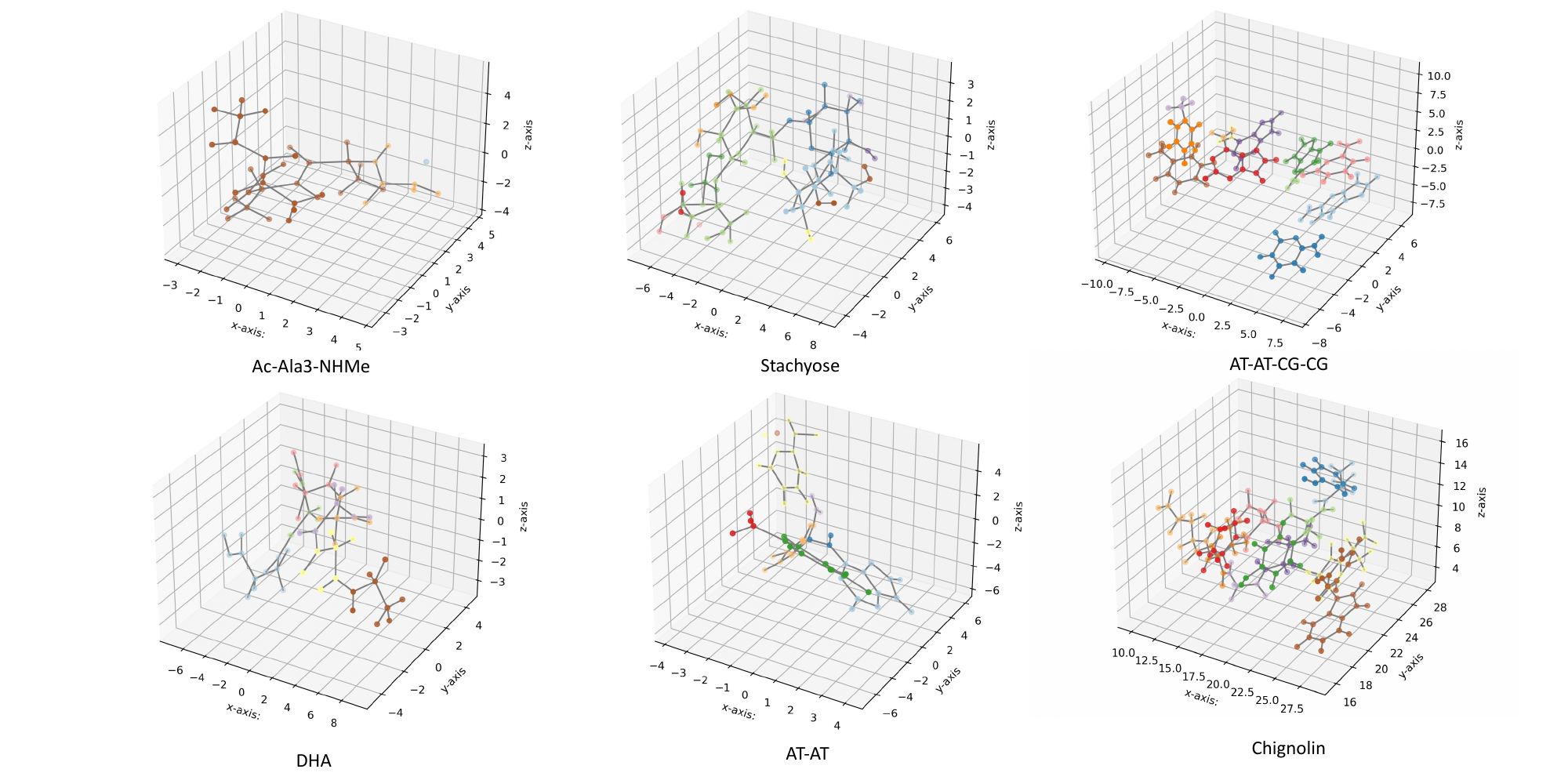}}
\caption{Visulization of BRICS fragmentation results}
\label{fig-brics}
\vskip -0.1in
\end{figure}

\section{Long-Range Module Visualization}
\label{sec:long-range-rep}
On top of the derivations in the main text, we provide a visualization of the long-range module which is implemented in ViSNet-LSRM, which is shown in Figure \ref{fig:long-range}.
\begin{figure}[ht]
\vskip 0.1in
\centering
\centerline{\includegraphics[width= \textwidth,trim=25 125 80 135,clip]{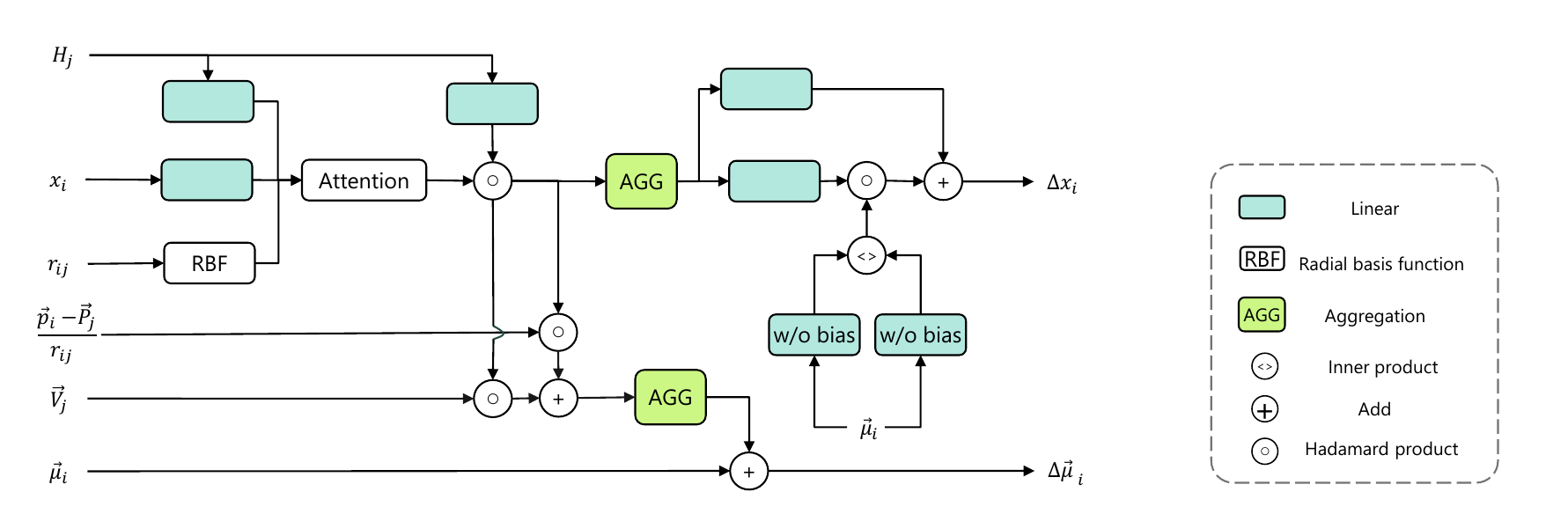}}
\caption{Structure of long-range module.}
\label{fig:long-range}
\vskip -0.1in
\end{figure}

\section{Proof of Equivariance}
\begin{myDef}
\text{(Rotation Invariance).} $f: \mathcal{X} \rightarrow \mathcal{Y}$ is rotation-invariant if $\forall R \in SO(3), X \in \mathcal{X}, f(XR) = f(X)$.
\end{myDef}

\begin{myDef}
\text{(Rotation Equivariance).} $f: \mathcal{X} \rightarrow \mathcal{Y}$ is rotation-equivariant if $\forall R \in SO(3), X \in \mathcal{X}, T \in \mathcal{T}, f(XR) = T(f(X))$.
\end{myDef}

\begin{proposition}
    The rotation invariance of the LSR-MP is preserved if the rotation invariance of short-range message-passing and long-range bipartite message-passing is preserved. 
\end{proposition}

\begin{proof}
     The goal is to show that if a rotation transformation $R \in SO(3)$ is applied to $\vec{p}$, the output $h_\text{out}$ remains unchanged.

     First, consider the short-range message-passing, which is rotationally invariant. This implies (Equation 1) that when a rotation transformation $R$ is applied to $\vec{p}$, the output $h^{\prime}$ remains the same as the original output $h$:

     \begin{equation}
        h^{\prime} = h = \textsc{ShortRangeModule}(Z, R\Vec{p}).
     \end{equation}
    
     The fragmentation learning module combines the resulting embeddings linearly. The output $H_j^{\prime}$ remains unchanged after applying the rotation transformation, as shown by the equality between $H_j^{\prime}$ and $H_j$:

      \begin{equation}
              H_j^{\prime} = \sum_{i \in S(j)} \alpha_i \odot h_i^{\prime} =  \sum_{i \in S(j)} \alpha_i \odot h_i = H_j.
          \label{eq:short-app}
      \end{equation} 
      
      Next, consider the long-range bipartite message-passing, which is also rotationally invariant. When a rotation transformation $R$ is applied to $\vec{\mu}$, $\vec{p}$, $\vec{V}$, and $\vec{P}$, the output $x$ remains unchanged:

      \begin{equation}
          x = \textsc{LongRangeModule}(h, H, x, R\vec{\mu}, R\vec{p}, R\vec{V},  R\vec{P}).
          \label{eq:long-app}
      \end{equation}

      By combining Eq \ref{eq:short-app} and Eq \ref{eq:long-app}, the transformed long-range scalar embeddings can be expressed as ${x}^{\prime} = \mu$:

      \begin{equation}
          {x}^{\prime} = \mu.
      \end{equation}

The final output $h_{\text{out}}^{\prime}$ remains unchanged after applying the rotation transformation to the input features. The Dense function combines the outputs from the short-range and long-range message-passing modules:

   \begin{equation}
 h_{\text{out}}^{\prime} = \textsc{Dense}\left(\left[h^{\prime}, x^{\prime}\right]\right) = \textsc{Dense}\left(\left[h, x\right]\right) =  h_{\text{out}}.
\end{equation}
\end{proof}

\begin{proposition}
The rotation equivariance of the LSR-MP is preserved if the rotation equivariance of short-range message-passing and long-range bipartite message-passing is preserved.
\end{proposition}

\begin{proof}
This proposition is equivalent to: If a rotation transformation $R \in SO(3)$ is applied to $\vec{p}$, one could find a predictable transformation $T \in \mathcal{T}$ to the vectorial embeddings $\vec{v}_\text{out}$.

Since short-range message-passing is rotationally equivariant, we can obtain the transformed short-range vectorial embeddings $\vec{v}^{\prime}$ by applying the rotation transformation $R$ to the input positions $\vec{p}$ and then passing them through the ShortRangeModule:

\begin{equation}
\vec{v}^{\prime} = T^{\text{short}}(\vec{v}) = \textsc{ShortRangeModule}(Z, R\Vec{p}).
\end{equation}

In the context of vectorial embeddings, the transformed short-range vectorial embeddings $\vec{v}^{\prime}$ are equal to the rotation transformation $R$ applied to the original vectorial embeddings $\vec{v}$:

\begin{equation}
\vec{v}^{\prime} = R\vec{v} = \textsc{ShortRangeModule}(Z, R\Vec{p}).
\end{equation}

The fragmentation learning module is composed of linear combinations of the resulting embeddings. This yields the transformed fragmentation learning module embeddings $\vec{V}_j^{\prime}$ and $\vec{P}_j^{\prime}$, which can be obtained by applying the rotation transformation $R$ to the original embeddings $\vec{V}_j$ and $\vec{P}_j$:

          \begin{align}
          \label{eq:37}
              \vec{V}_j^{\prime} &= \sum_{i \in S(j)} \beta_i \odot \vec{v}_i^{\prime} =  \sum_{i \in S(j)} \beta_i \odot R\vec{v}_i = R\sum_{i \in S(j)} \beta_i \odot \vec{v}_i = R\vec{V}_j \\
              \vec{P}^{\prime}_j &= \sum_{i \in S(j)} \gamma_i  \vec{p}_i^{\prime} =  \sum_{i \in S(j)} \gamma_i R\vec{p}_i = R\sum_{i \in S(j)} \gamma_i \vec{v}_i = R\vec{P}_j .
          \end{align}

As the long-range bipartite message-passing is also equivariant, this indicates that the transformed long-range vectorial embeddings could be obtained via $R\vec{\mu}$:

\begin{equation}
R\vec{\mu} = \textsc{LongRangeModule}(h, H, x, R\vec{\mu}, R\vec{p}, R\vec{V},  R\vec{P}).
\label{eq:38}
\end{equation}

Combing Eq. \ref{eq:37} and Eq. \ref{eq:38}, we have:

\begin{equation}
\mu^{\prime} = R\mu = \textsc{LongRangeModule}(h, H, x, R\mu, R\vec{p}, R\vec{V},  R\vec{P}).
\end{equation}

This equation shows that the transformed long-range vectorial embeddings $\mu^{\prime}$ are equal to the rotation transformation $R$ applied to the original vectorial embeddings $\mu$.

Finally, the output vectorial embeddings are given by:

\begin{equation}
\vec{v}_{\text{out}}^{\prime} = U(R[\vec{v}, \vec{\mu}]) = R\vec{v}_{\text{out}}.
\end{equation}

This equation demonstrates that the transformed output vectorial embeddings $\vec{v}_{\text{out}}^{\prime}$ can be obtained by applying the rotation transformation $R$ to the concatenated original vectorial embeddings $\vec{v}$ and $\vec{\mu}$, and then passing them through the linear combination function $U$. The result is equal to the rotation transformation $R$ applied to the original output vectorial embeddings $\vec{v}_{\text{out}}$.

In summary, the proof shows that if the rotation equivariance of short-range message-passing and long-range bipartite message-passing is preserved, the rotation equivariance of the LSR-MP is also preserved.
\end{proof}

\section{Loss Function}
In our study, we used a combination of mean squared error (MSE) loss functions for energy and force to train our models. Specifically, we minimized the weighted sum of MSE between the predicted and true energy values and force values during training. The weight for energy was set to $1-\rho$, while the weight for force was set to $\rho$. 
\begin{equation}
L = (1-\rho)\frac{1}{n} \sum_{i=1}^{n}(E_i - \hat{E_i})^2 + \frac{\rho}{3n} \sum_{i=1}^{n}\lVert\mathbf{F}_i - \hat{\mathbf{F}_i}\rVert^2
\end{equation}

where $n$ is the number of samples, $E_i$ is the true energy value for sample $i$, $\hat{E_i}$ is the predicted energy value for sample $i$, $\mathbf{F}_i$ is the true force value for sample $i$, and $\hat{\mathbf{F}_i}$ is the predicted force value for sample $i$. The trade-off parameter $\rho$ controls the relative importance of minimizing energy versus force errors during training. The $\lVert\cdot\rVert$ symbol represents the $L2$ norm, which calculates the Euclidean distance between the true and predicted force vectors, providing a measure of the overall force prediction error.

\section{Hyperparameter setting}
\subsection{Hyperparameter in Table \ref{table:main}}
\label{app-baslines}

For baseline methods that provided an official implementation (Allegro, PaiNN, ET), we used their code directly, while for methods that are not open source (ViSNet), we reimplemented them based on their papers. All models were trained using NVIDIA A6000 GPUs. For training on MD22, we used kcal/mol as the default unit for regression. For training on Chignolin, we used eV. 
\begin{table}[htbp]
\centering
\caption{Comparison of Hyperparameters used for MD22 and Chignolin. NA indicates the hyperparameter is not used.}
\label{tab:hyper}
\resizebox{\textwidth}{!}{
\begin{tabular}{lccccccc}
\hline
                     & ViSNet-LSRM & Equiformer LSRM           & PaiNN    & ET       & ViSNet   & Allegro     & Equiformer \\ \hline
batch size           & 8           & 8           & 32       & 32       & 8        & \{5, 3, 1\} & 8 \\
$l_{max}$            & 1           & 2          & 1        & 1        & 1        & 3           & 2 \\
hidden channels      & 128         & 128($l$=0), 64($l$=1), 32($l$=2)           & 128      & 128      & 128      & 128        & 128($l$=0), 64($l$=1), 32($l$=2) \\
learning rate        & 5.00E-04    & 5.00E-04           & 1.00E-03 & 1.00E-03 & 5.00E-04 & 2.00E-03    & 5.00E-04\\

(short-range) layers & 4, 6        & 4             & 6        & 6        & 6        & 3           & 6 \\
long-range layers    & 2           & 2          & NA       & NA       & NA       & NA          & NA \\
Warm up steps        & 1000        & 1000           & 1000     & 1000     & 1000     & 0           & 1000 \\
early stop patience  & 500         & 500           & 500      & 500      & 500      & 1000        & 500 \\
lr scheduler         & Cosine Annealing & Cosine Annealing     & ReduceLROnPlateau         & ReduceLROnPlateau         & Cosine Annealing & ReduceLROnPlateau    & Cosine Annealing \\
min lr               & 1.00E-07    & 1.00E-07           & 1.00E-07 & 1.00E-07 & 1.00E-07 & 1.00E-06    & 1.00E-07\\
energy/force weight  & 1 / 80      & 1/80           & 1 / 99   & 1 / 99   & 1 / 80   & 1 / 1000    & 1 / 80\\
(short) cutoff       & 4           & 4           & 4        & 4        & 4        & \{7, 4\}    & 4 \\
long cutoff          & 9           & 9           & NA       & NA       & NA       & NA           & NA\\
\hline
\end{tabular}
}
\end{table}

\subsection{Hyperparameter in Table~\ref{table: number-of-parameters}}
\label{app:detail-table}
\begin{table}[htbp]
\vspace{-3mm}
\centering
\caption{Comparison of the number of parameters and training speed of various methods when forces MAE is comparable on the molecule AT-AT-CG-CG.}
\label{table: number-of-parameters-detail}
\resizebox{\linewidth}{!}{
\begin{tabular}{lccccccc}
\toprule
Methods (MAE)        & ViSNet (0.16) & ViSNet-LSRM (0.13)     & PaiNN (0.35) & ET (0.29) & Allegro (0.13) & Equiformer (0.13)\\ \midrule
\# of Parameters & 2.21M & \textbf{1.70M} & 3.20M & 3.46M & 15.11M & 3.02M\\
Layer Number & 8 & 4+2 & 8 & 8 & 3 & 4 \\
hidden channels & 128 & 64($l=0$), 48 ($l=1$) & 128 & 128 & 128 & 128($l=0$), 64 ($l=1$), 32($l=0$) \\
Training Time / Epoch (s) & 44 & \textbf{19} & \textbf{19} & 26  & 818 & 155 \\ \bottomrule
\end{tabular}
}
\vspace{-2mm}
\end{table}

\section{Higher-order LSRM}
\label{higher-order}
In this section, we detail the application of LSR-MP to another series of EGNNs that employ the Clebsh-Gorden (CG) tensor product in its core architecture. In particular, we chose the state-of-the-art Equiformer as the short-range model. For simplicity, we would ignore parity in the ensuing discussion.

\subsection{Preliminary}

A group representation characterizes how group elements, such as rotations and translations, act on a vector space. In the 3D Euclidean group $E(3)$ context, we consider scalars and Euclidean vectors in $\mathbb{R}^3$. Scalars remain unchanged under rotation, while Euclidean vectors transform accordingly. To address translation symmetry, we work with relative positions.

    \textbf{Irreducible representations (irreps)} are the minimal building blocks of group representations. They consist of transformation matrices acting on specific vector spaces. For each $g\in SO(3)$ group, representing 3D rotations, we have irreps matrices $D_L(g) \in R^{{(2L+1)} \times {(2L+1)}} $, known as Wigner-D matrices. These matrices act on vector spaces of dimension $(2L + 1)$, where $L$ is a non-negative integer. Vectors transformed by $D_L(g)$ are type-$L$ vectors. In EGNNs, each of these vectors could serve as a hidden neuron, and neurons with $L \geq 1$ are normally referred to as vector neurons. We can concatenate multiple type-$L$ vectors to construct $SE(3)$-equivariant irreps features.   For example, scalar features $h \in \mathbb{R}^{1 \times d_0}$ in the main text are composed of type-$0$ vectors, with 1 order, and $d_0$ channels. Vectorial features $\vec{v} \in\mathbb{R}^{3 \times d_1}$  are composed of type-$1$ vectors, with $3$ distinct orders, and $d_1$ channels. In general, type-$L$ features $f_{L} \in \mathbb{R}^{(2L + 1) \times d_L}$ have $2L + 1$ distinct orders and $d_L$ channels.
    The irreducible representation could be written as a data structure that aggregates irreps of different types:

    \begin{equation}
        f = \{f_0 \in \mathbb{R}^{1\times d_0}, f_1 \in \mathbb{R}^{3\times d_1}, \cdots, f_L \in \mathbb{R}^{(2L+1)\times d_L}\}.
    \end{equation}

   \textbf{Spherical harmonics (SH)} are functions capable of projecting Euclidean vectors in $\vec{r} \in \mathbb{R}^3$ into type-$L$ vectors:
   \begin{equation}
       SH(\cdot): \vec{r} \in \mathbb{R}^{3} \rightarrow f_L\in \mathbb{R}^{2L+1} = Y^L(\frac{\vec{r}}{||\vec{r}||}),
   \end{equation}
which exhibit $E(3)$-equivariance, preserving the group structure during vector transformations:
    \begin{equation}
        \forall g \in SO(3), Y^L(\frac{D^L(g)\vec{r}}{||D^L(g)\vec{r}||}) = D^L(g)Y^L(\frac{\vec{r}}{||\vec{r}||}).
    \end{equation}
  By employing SH of relative positions between node-$i$ and node-$j$, we could generate the initial set of irreps features: 
  \begin{equation}
      (f_L)^0_{ij} = Y^L(\frac{\vec{p}_{ij}}{||\vec{p}_{ij}||}).
  \end{equation}
  
  Equivariant information propagates through irreps features via operations such as tensor-product-based directional message passing.

  \textbf{Clebsh-Gorden Tensor Product (CG product)} is used form interactions between two different vectors of type-$L_1$ and type-$L_2$, and output a vector of type -$L_3$:
\begin{equation}
    f_{L_3}= f_{L_1}  \otimes  f_{L_2},
\end{equation}
where $\otimes$ denotes CG product. 
The tensor product utilizes Clebsch-Gordan coefficients, to combine vectors of different types and different orders:
\begin{equation}
    f_{(L_3, m_3)} = (f_{L_1}  \otimes  f_{L_2})_{m_3} = \sum_{m_1 = - L_1}^{L_1}\sum_{m_2 = - L_2}^{L_2}C^{(L_3, m_3)}_{(L_1, m_1), (L_2, m_2)} f_{(L_1, m_1)} f_{(L_2, m_2)},
\end{equation}

where $m$ indexes to the order of a type-$L$ tensor, and $C^{(L_3, m_3)}_{(L_1, m_1), (L_2, m_2)}$ is the Clebsch-Gordan coefficient.
CG coefficient are non-zeros only when:
\begin{equation}
   |L_1 - L_2|\leq L_3 \leq |L_1 + L_2|,
\end{equation}
which restricts the output type of the operations. For example, type-$1$ vectors and type-$1$ vectors could not produce type-$3$ vectors,  since the output types are restricted in $\{0,1, 2\}$. 

 \textbf{Parameterized CG product}: The tensor product could be formulated as a parameterized operation where a learnable parameter $W^{L_3}_{L_1, L_2}$ is assigned to each combination of $(L_1, L_2, L_3)$. Each type-$L$ features are consist of $d$ channels of type-$L$ vectors, thus the learnable parameters could be generalized to the following forms:
 \begin{equation}
     (L_1, d_1, L_2, d_2, L_3, d_3) \leftrightarrow\ W^{(L_3, d_3)}_{(L_1,d_1), (L_2,d_2)},
 \end{equation}
 where $d_1$, $d_2$ and $d_3$ corresponds to the channel index of the input $f_1$, $f_2$ and the output $f_3$. This leads to the parameterized CG product which could be written as:
 \begin{equation}
     f_{L_3}= f_{L_1}  \otimes^{W}  f_{L_2}.
 \end{equation}

\subsection{Notation}
\textbf{Notations:}
To distinguish short-range and long-range embeddings, we denote short-range embeddings as $h$, with $h$ being:
    \begin{equation}
        h = \{(h_0) \in \mathbb{R}^{1\times d_0}, (h_1) \in \mathbb{R}^{3\times d_1}, \cdots, (h_L) \in \mathbb{R}^{(2L+1)\times d_L}\}.
        \label{eq-form}
    \end{equation}

while long-range embeddings as $x$. The fragments embeddings are capitalized, which are denoted as $H$. The short-range message is denoted as $m_{ij}$ and the long-range message is denoted as $M_{ij}$. $H$, $x$, $m_{ij}$, $M_{ij}$ are all irreducible representations, taking the same form as Eq. \ref{eq-form}. In $(h_L)^l_i$, $L$ indexes to the type of the irreps, $l$ indexes to the layer number in a multilayer framework, and $i$ indexes to different atoms. Meanwhile, we also used a short-hand notation $h^l_i$, where the index of the type of the irreps is ignored.

\subsection{Short-Range Module}

Similar to the short-range module  introduced in the main text, the short-range module performs message passing on $\mathcal{G}_{\text{short}}$ by taking the atomic numbers $Z \in \mathbb{N}^{n \times 1}$ and positions $\vec{p} \in \mathbb{R}^{n \times3}$ ($n$ is the number of atoms in the system) as input, and model the geometric information on $\mathcal{G_{\text{short}}}$. 
 
 For type-$0$ vectors, they are initialized with atom number embeddings:
\begin{equation}
    (h_0)^0_{i}= \textsc{Embed}(z_i),
\end{equation}

where $z_i\in Z$ is the atomic number, $\textsc{Embed}(\cdot): \mathbb{N} \mapsto \mathbb{R}^{d_0}$, is a learnable embedding map, with $d_0$ being the number of hidden channels for type-$0$ vectors.

For type-$L$ vectors with $ L \geq 1$, they are initialized to be zeros:
\begin{equation}
     (h_L)^0_{i}= \mathbf{0} \in \mathbb{R}^{(2L+1) \times d_L}, \text{if\ } L \geq 1,
\end{equation}
where a type-$L$ irreps are composed of $d_L$ channels of type-$L$ vectors.

In general, the short-range module adopts an iterative short-range message-passing scheme:
{\small
\begin{equation}
        h_i^{l} =  \phi_{\text{Short}}\left(h_i^{l-1}, \sum_{j\in N(i)} m
_{ij}^{l-1} \right). 
\end{equation}
}The $\phi_{\text{Short}}(\cdot)$ defines a message-passing framework, and $N(\cdot)$ denotes the first-order neighbor set of a particular node. $m_{ij}$  denote the message message between node $i$ and its first order neighbor $j$. $m_{ij}$ is composed of irreps of different types and is computed using the following message functions:
\begin{eqnarray}
    m_{ij} = \phi(h_i, h_j, \vec{p}_{ij}).
\end{eqnarray}
Commonly, spherical harmonics are used to turn $\vec{p}_{ij}$ into irreducible representations to interact with the node irreps:
\begin{eqnarray}
    m_{ij} = \phi(h_i, h_j, SH(\vec{p}_{ij}), r_{ij}),
\end{eqnarray}
where $SH(\cdot)$ is the spherical harmonics, and $r_{ij}$ is the distance from node-$i$ to node $j$.
This interaction could be performed via a parameterized CG product:

\begin{equation}
    m_{ij} = h_i \otimes^{W(r_{ij})} SH(\vec{p}_{ij}).
\end{equation}

 \subsection{Fragmentation Module}
To obtain the type-$L$ irreps  for a given fragment, the type-$L$  
 irreps of the contained node representation were summed:
\begin{eqnarray}
     \label{eq-app:FLfunction1}
    (H_L)_j^{l} &=& \sum_{i\in S(j)}  (\alpha_L)_i^{l} \odot (h_L)_i^{l},
\end{eqnarray}
in which $(H_L)_j^{l}$ denotes the type-$L$ irreps of fragment $j$,  $j$ is the index for fragments, and $S(j)$ is the set induced by the assignments of fragment $j$.  $\alpha_i^{l}$ are weight vectors for each atom within the fragments. 

\subsection{Long-Range Module}
\label{sec:Long-Range Module-2}

The long-range module is targeted to capture possible atom-fragment interactions on $\mathcal{G}_{\text{long}}$. 
Generally, the long-range embeddings are initialized based on the short-range embeddings of layer-$L_\text{short}$. Type-$0$ vectors are initialized with $\textsc{Dense}(\cdot)$ acting on type-$0$ irreps of the short-range representation. For $L\geq1$, type-$L$ vectors are initialized with $U(\cdot)$ (linear layer without bias) acting on type-$L$ irreps of the short-range representation:

\begin{equation}
(x_L)^0_i = \begin{cases}
 \textsc{Dense}\left((h_{0})_i^{L_\text{short}}\right), & \text{if\ } L = 0,\\
U\left((h_{L})_i^{L_\text{short}}\right), 
   &\text{if\ } L \geq 1,
\end{cases}
\end{equation}
    
The geometric message-passing is performed to characterize long-range interactions:
{\small
\begin{equation}
        x^l_i =  \psi_{\text{long}}\left(x_i^{l-1}, \sum_{j\in N(i)} M
_{ij}^{l-1}\right). 
\end{equation}
}$\psi(\cdot)_{\text{long}}$ is the general bipartite message passing framework, $x^l_i$ is the long-range irreps, $i,j$ are index for atom, and $N(\cdot)$ is the neighborhood of atom $i$ on the atom-fragment bipartite graph. $M_{ij}$, denoting  the  message between atom $i$ and its incident fragment $j$, is comprised of irreps of different types and are obtained via the message functions $\psi_{}(\cdot)$:
\begin{align}
    M_{ij}& = \psi\left(h_i, H_j, \vec{p}_i, \vec{P}_j\right),
\end{align}

Similarly, spherical harmonics are commonly used to form the interaction between $h_i$, $H_j$:
\begin{align}
    M_{ij}& = \psi\left(h_i, H_j, SH(\vec{p}_i - \vec{P}_j), ||\vec{p}_i - \vec{P}_j|| \right). 
\end{align}

\subsection{Properties Prediction}
A fusion strategy can be applied to integrate the long-range irreps and short-range irreps:

\begin{equation}
(x_L)_\text{out} = \begin{cases}
 \textsc{Dense}\left(\left[(h_0)^{L_\text{short}}, (x_0)^{L_\text{long}}\right]\right), & \text{\ if\ } L = 0,\\
U\left(\left[(h_L)^{L_\text{short}}, (x_L)^{L_\text{long}}\right]\right), & \text{\ if\ } L\geq 1;
\end{cases}
\end{equation}
$L_\text{short}$ and $L_\text{long}$ are the number of layers for the short-range module and long-range module respectively.

\subsection{Equiformer-LSRM}
Based on LSR-MP framework and irreducible representation, we provided another exemplary implementation of LSR-MP called Equiformer-LSRM. It uses Equiformer~\cite{ liao2023equiformer} as the backbone for the short-range module, thus $h^{l} =  \textsc{Equiformer}(Z, \vec{p}).$
In the fragment learning module, we choose $\gamma_i$ to be $\frac{z_i}{\sum_{i, i\in S(c)}  z_i}$, i.e. setting the center of any fragment to be the center of mass of all contained atoms. $\alpha_i$ are chosen to be $\mathbf{1^d}$. For the long-range module, we reused the \textsc{LayerNorm} and  \textsc{DPTransBlock}  implemented in Equiformer~\cite{ liao2023equiformer} to pass messages in the atom-fragment bipartite graph. In particular, we first perform layer normalization on the fragment irreducible representation and we concatenated the atom representations and fragment representations to form a new set of node representations:
\begin{equation}
    h^\prime = \left[ h, H \right].
\end{equation}

The atom positions and fragment positions were also concatenated:
\begin{equation}
    \vec{p}^\prime = \left[ \vec{p}, \vec{P} \right].
\end{equation}

Message passing were performed on the atom-fragment bipartite graph, thus:

\begin{equation}
    x =  \textsc{DPTransBlock}(h^\prime, \mathcal{G}_{\text{long}}, \vec{p}^\prime).
\end{equation}

\end{appendices}

\end{document}

%% file: math_commands.tex
\usepackage{amsmath,amsfonts,bm}

\def\eqref#1{equation~\ref{#1}}

\def\1{\bm{1}}

\DeclareMathAlphabet{\mathsfit}{\encodingdefault}{\sfdefault}{m}{sl}
\SetMathAlphabet{\mathsfit}{bold}{\encodingdefault}{\sfdefault}{bx}{n}

